\documentclass{llncs}

\usepackage{amsmath}
\usepackage{stmaryrd}
\usepackage{amsfonts}
\usepackage{amssymb}
\usepackage{tikz}
\usepackage{graphics}
\usepackage{graphicx}
\usepackage{algorithm}
\usepackage{algorithmic}
\usepackage{url}
\usepackage{multirow}
\usepackage[margin=1in]{geometry}

\newcommand{\qedsymbol}{\hfill \rule{0.7em}{0.7em}}

\newtheorem{concentration conditions}{Concentration Conditions}
\newtheorem{fact}{Fact}

\def\Inf{{\it Inf}}
\def\Var{{\it Var}}
\newcommand{\argmax}{\operatornamewithlimits{argmax}}
\newcommand{\argmin}{\operatornamewithlimits{argmin}}
\def\ComputeProbability{{\sf CompProb}}
\def\MCCP{\mbox{\sf MC-CompProb}}
\def\BiCP{\mbox{\sf Bi-CompProb}}
\def\MinSeed{\mbox{\sf MinSeed-PCG}}

\def\prob{{\it prob}}

\newcommand{\rmv}[1]{}

\begin{document}
%

\title{ 
Minimizing Seed Set Selection with Probabilistic Coverage Guarantee
	in a Social Network}
\author{Peng Zhang\inst{1} \and Wei Chen\inst{2} \and Xiaoming Sun\inst{3}\thanks{This work was supported in part by the National Natural Science Foundation of China Grant 61170062, 61222202 and the National Program for support of Top-notch Young Professionals.}
\and Yajun Wang\inst{2} \and Jialin Zhang\inst{3}}
\institute{Purdue University \and Microsoft \and Institute of Computing Technology, Chinese Academy of Sciences}

\sloppy

\maketitle
\begin{abstract}
A topic propagating in a social network 
	reaches its tipping point if
	the number of users discussing it in the network exceeds a critical threshold such that
	a wide cascade on the topic is likely to occur.
In this paper, we consider the task of selecting initial seed users of a topic
	with minimum size
	so that {\em with a guaranteed probability} the number of users
	discussing the topic would reach a given threshold.
We formulate the task as an optimization problem called
	{\em seed minimization with probabilistic
	coverage guarantee (SM-PCG)}.
This problem departs from the previous studies on social influence maximization or seed minimization
	because it considers influence coverage with {\em probabilistic} guarantees instead of guarantees on {\em expected} influence coverage.
We show that the problem is not submodular, and thus is harder than previously studied problems based
	on submodular function optimization.
We provide an approximation algorithm and show that it approximates the optimal
	solution with both a multiplicative ratio and an additive error.
The multiplicative ratio is tight while the additive error would be small
	if influence coverage distributions of certain seed sets are well concentrated.
For one-way bipartite graphs we analytically prove the concentration condition and obtain an approximation algorithm
	with an $O(\log n)$ multiplicative ratio and an $O(\sqrt{n})$ additive error, where $n$ is the total number of nodes in the social graph.
Moreover, we empirically verify the concentration condition in
	real-world networks and experimentally demonstrate
	the effectiveness of our proposed algorithm comparing to commonly adopted benchmark algorithms.
\end{abstract}




\section{Introduction}
With online social networks such as Facebook and Twitter becoming popular for people to express
their thoughts and ideas, or to chat with each other, online social networks provide a platform for triggering a hot topic and then influencing a large population.
Different from most traditional media (such as TV and newspapers), information spread on social networks mainly base on the trust relationship between individuals. Consider the following scenario: when someone publishes a topic on the online social network, his/her friends will see this topic on the website. If they think it is interesting or meaningful, they may write some comments to follow it or just forward it on the website as a response. Similarly, the comments or forwarding from these friends will attract their own friends,
	leading to more and more people on the social network paying attention to that topic.
When the number of users discussing about this topic on the online social network
	reaches certain critical threshold, this topic becomes a {\em hot topic}, which
	is likely to be surfaced at the prominent place on the social networking
	site (e.g. 10 hot topics of today), and is likely to be picked up by
	traditional media and influential celebrities.
In turn this will generate an even wider cascade causing more people to discuss
	about this topic.

Therefore, making a topic reach the critical threshold
	(also called the {\em tipping point}~\cite{Gladwell02})
	is the crucial step
	to generate huge influence on the topic, which is desirable by companies
	large and small trying to use social networks to promote their products,
	through the so called {\em viral marketing campaigns}.
Besides making the content of the topic attractive and viral, another
	key aspect is to select seed users in the network that initiate the topic
	discussion effectively to trigger a large cascade on the topic.
Due to the cost incurred for engaging seed users (e.g. providing free sample
	products), it is desirable that the size of seed users is minimized.
Moreover, the marketers also need certain probabilistic guarantee on how
	likely the viral marketing campaign could reach the desired critical threshold
	in order to trigger an even larger cascade via hot topic listings,
	traditional media coverages, and celebrity followings.
Hence, the problem at hand is how to select a seed set of users of minimum size
	to trigger a topic cascade such that the cascade size reaches the
	desired critical threshold with guaranteed probability.

In this paper, we formulate the above problem as the following optimization
	problem and call it {\em seed minimization with probabilistic coverage
	guarantee (SM-PCG)}.
A social network is modeled as a directed graph, where nodes represent individuals and directed edges represent the relationships between pairs of individuals.
Each edge is associated with an {\em influence probability}, which means that once a node is activated, it can activate its out-neighbors through the outgoing edges with their
	associated probabilities at the next step.
Our analytical results work for a large class of influence diffusion models that guarantee
	submodularity (the diminishing marginal return property in terms of seed set size),
	but for
	illustration purpose, we adopt the classic {\em independent cascade (IC)
	model}~\cite{kempe2003maximizing} as the influence diffusion model.
In the IC model, initially all seed nodes are activated while others are
	inactive, and at each step, nodes activated at the previous step have
	one chance to activate each of its inactive out-neighbors in the network.
The total number of active nodes after the diffusion process ends is referred as
	the {\em influence coverage} of the initial seed set.
Given such a social network with influence probabilities on edges,
	given a required coverage threshold $\eta$ and a probability threshold $P$,
	the SM-PCG problem is to find a seed set $S^*$ of minimum size
	such that the probability that the influence coverage of $S^*$
	reaches $\eta$ or beyond is at least $P$.

The formulation of the SM-PCG problem significantly departs from previous optimization
	problems based on social influence diffusion
	(e.g.~\cite{kempe2003maximizing,ChenWY09,chen2010scalable,goyal2012minimizing})
	in that it requires the selected seed set to satisfy a {\em probabilistic} coverage
	guarantee, while previous research focuses on {\em expected} coverage guarantee.
For the application of generating a hot topic, we believe that
	it is reasonable to ask for a guarantee on the probability of influence coverage exceeding
	a given threshold, since this provides direct information on the likelihood of
	success of the viral marketing campaign, which is very helpful for marketers to gauge their cost and benefit trade-offs for the campaign.
Merely saying that the expected influence coverage exceeds the required coverage threshold
	is not enough in this case.
To the best of our knowledge, this is the first work that focuses on probabilistic
	influence coverage guarantee among existing studies on social network influence
	optimization problems.

In this paper, we first show that the set functions based on the SM-PCG problem
	are not submodular, which means that it is more difficult than most of the existing
	social influence optimization problems that rely on submodular set function
	optimizations.
Next, we investigate two computation tasks related to SM-PCG problem, one is
	to fix a seed set $S$ and a coverage threshold $\eta$ and compute the probability of
	influence coverage of $S$ exceeding $\eta$, and the other
	is to fix a seed set $S$ and a probability threshold $P$, and compute the maximum
	coverage threshold $\eta$ such that the probability of influence coverage
	of $S$ exceeding $\eta$ is at least $P$.
We show that the first problem is \#P-hard but can be accurately estimated, while
	the second one is \#P-hard to even approximate the value within any nontrivial ratio.
These results further demonstrate the hardness of the problem.

We then adapt the greedy approximation algorithm targeted for expected influence coverage
	problem (which is submodular) to the SM-PCG problem.
Although the adapted algorithm still follows the greedy approach,
	our main contribution is on a detailed analysis,
	which proves that our algorithm approximates the
	optimal solution with both a multiplicative ratio and an additive error.
The multiplicative ratio is due to the greedy approximation algorithm for
	expected influence coverage and is tight, while the additive error is determined
	by the concentration property (in particular the standard deviations)
	of influence coverage distributions of two specific seed sets.
For one-way bipartite graphs where edges are directed from one side to the other side,
	we analytically show that the influence coverage distributions are well concentrated
	and we could reach an additive error of $O(\sqrt{n})$ where $n$ is the total
	number of nodes in the graph.

Finally, using several real-world social networks including a network with influence probability parameters
	obtained from prior work, we empirically validate our approach
	by showing that (a) influence coverage distributions of seed sets are well
	concentrated, and (b) our algorithm selects seed sets with sizes much smaller
	than commonly adopted benchmark algorithms.

To summarize, our contributions include: (a) we propose the study of seed minimization with probabilistic
	coverage guarantee (SM-PCG), which is more relevant to hot topic generation in online social networks and has
	not been studied before;
	(b) we show that neither of the two versions of set functions related to SM-PCG is submodular,
		one version is \#P-hard to compute but allows accurate estimation while the other
		version is \#P-hard to even approximate to any nontrivial ratio;
	(c) we adapt the greedy algorithm targeted for expected coverage guarantee to SM-PCG, and analytically
	show that the adapted algorithm provides an approximation guarantee with a tight multiplicative ratio
	and an additive error depending on the influence coverage concentrations of certain seed sets; and
	(d) we empirically demonstrate the effectiveness of our algorithm using real-world datasets.
	
\subsection{Related Work}
{\em Influence maximization}, as the dual problem of seed minimization, is to find
	a seed set of at most $k$ nodes to maximize the expected influence coverage
	of the seed set.
Domingos and Richardson are the first to formulate influence maximization problem from an
	algorithmic perspective~\cite{domingos2001mining,richardson2002mining}. Kempe et al. first model this problem as a discrete optimization problem \cite{kempe2003maximizing},
	provide the now classic independent cascade and linear threshold diffusion models,
	and establish the optimization framework based on submodular set function optimization.
A number of studies follow this approach and provide more efficient
	influence maximization algorithms
	(e.g.~\cite{ChenWY09,chen2010scalable,ChenYZ10,simpath}).
In~\cite{long2011minimizing}, Long et al. first study independent cascade and linear threshold diffusion models from a minimization perspective.
In~\cite{goyal2012minimizing}, Goyal et al. provide a bicriteria approximation algorithm
	to minimize the size of the seed set with its expected influence coverage
	reaching a given threshold.
Recently, a continuous time diffusion model is proposed and studied in ~\cite{rodriguez2012influence} 	and ~\cite{du2013scalable}.
All these existing studies focus on expected influence coverage, and rely on
	the submodularity of expected influence coverage function for the optimization task.
In contrast, we are the first to address probabilistic coverage guarantee for
	the seed minimization problem, which is not submodular.

Seed minimization with non-submodular influence coverage
	functions under different diffusion models
	have been studied.
Chen~\cite{chen2009approximability} studies the seed minimization problem
	under the fixed threshold model, where a node is activated when its active neighbors
	exceed its fixed threshold.
He shows that the problem cannot be approximated within any polylogarithmic factor
	(under certain complexity theory assumption).
Goldberg and Liu~\cite{GoldbergL13} study another variant of fixed threshold model
	and provide an approximation algorithm based on the linear programming technique.
Influence coverage functions in both models are deterministic and non-submodular.
However, these models are quite different from the model we study in this paper, and
	thus their results and techniques are not applicable to our problem.

The rest of this paper is organized as follows. We define the diffusion model
	and the optimization problem SM-PCG in Section~\ref{sec:model}, and provide
	related results and tools in Section~\ref{sec:tools}, including
	the non-submodularity of the set functions for SM-PCG.
In Section~\ref{sec:infcomp} we investigate the computation problems related to
	SM-PCG.
In Section~\ref{sec:algo} we provide our algorithm for general graphs and
	analyze its approximation guarantee.
In Section~\ref{sec:bipartite} we provide algorithmic and analytical results for
	one-way bipartite graphs.
We empirically validate our concentration assumption on influence coverage distributions and
	the effectiveness of our algorithm in Section~\ref{sec:exp}, and
	conclude the paper in Section~\ref{sec:conclusion} with a discussion on potential
	future directions.
\rmv{Due to the space constraint, some of the technical proofs and empirical results are omitted,
	and they are included in our full technical report \cite{fullreport14}.}

\section{Model and Problem}
\label{sec:model}
In our problem, a social network is modeled as a directed {\em social graph} $G=(V,E)$, where
	$V$ is the set of $n$ vertices or nodes representing individuals in a social network,
	and $E$ is the set of directed edges representing influence relationships between
	pairs of individuals.
Each edge $(u,v)\in E$ is associated with an {\em influence probability} $p_{u,v}$.
Intuitively, $p_{u,v}$ is the probability that node $u$ activates node $v$ after $u$ is activated.
The influence diffusion process in the social graph $G$ follows the independent
	cascade (IC) model, a randomized process summarized in \cite{kempe2003maximizing}.
Each node has two states, {\em inactive} or {\em active}.
The influence diffusion proceeds in discrete time steps, and we say that a node $u$
	{\em is activated at time $t$} if $t$ is the first time step at which $u$ becomes active.
At the initial time step $t=0$, a subset of nodes $S\subseteq V$ is selected as active nodes,
	defined as the \emph{seed set}, while other nodes are inactive.
For any time $t\ge 1$, when a node $u$ is activated at step $t-1$, $u$ is given a single chance to activate each of its inactive out-neighbors $v$ through edge $(u,v)$ independently with probability $p_{u,v}$ at step $t$.
Once activated, a node stays as active in the remaining time steps.
The influence diffusion process stops when there is no new activation at a time step.

Given a target set $U\subseteq V$, let $\Inf_U(S)$ be the random variable denoting the number
	of active nodes in $U$ after the diffusion process starting from the seed set $S$ ends.
When the context is clear, we usually omit the subscript $U$ and use $\Inf(S)$ to represent
	this random variable, and we refer $\Inf(S)$ as the {\em influence coverage} of seed
	set $S$ (for target set $U$).
The optimization problem we are trying to solve is to find a seed set $S$ of minimum size such that
	the influence coverage of $S$ is at least a required threshold with a required probability
	guarantee.
The formal problem is defined below.

\begin{definition}[Seed minimization with probabilistic coverage guarantee]
We define the problem of {\em seed minimization with probabilistic coverage guarantee (SM-PCG)}
	as follows.
The input of the problem includes the social graph $G=(V,E)$, the influence probabilities
	$p_{u,v}$'s on edges, the target set $U$, a coverage threshold
	$\eta < |U|$,\footnote{We believe that $\eta < |U|$ is reasonable
	for the application scenarios we described since typically
	it requires only a fraction of the entire target node set to make a topic hot.
For the case of $\eta = |U|$, we also worked out a separate solution for
		one-way bipartite graphs, 
		and describe our algorithm in Appendix.
	}  a probability threshold $P \in (0,1)$.
The problem is to find the minimum size seed set $S^*$ such that $S^*$ can activate at least $\eta$ nodes in $U$ with probability $P$, that is,
\[S^* = \argmin_{S:\Pr(\Inf(S)\geq \eta )\geq P} |S|.
\]
\end{definition}
	

The following theorem shows the hardness of the SM-PCG problem.

\begin{theorem} \label{thm:hard}
The problem SM-PCG is NP-hard, and for any $\varepsilon > 0$, it cannot be approximated
	within a ratio of $(1-\varepsilon)\ln n$ unless NP has $n^{O(\log\log n)}$-time deterministic
	algorithms.
\end{theorem}
\begin{proof}
The problem of Set Cover is a special case of SM-PCG.
We can represent an instance of Set Cover as a bipartite graph $G=(U,V,E)$, where
	$U$ is the set of elements, and $V$ is the set of subsets of $U$, and an edge $(u,v)\in E$
	for $u\in U$ and $v\in V$ means $u$ is in the subset $v$.
The problem is to find a subset $S \subseteq V$ of minimize size such that all elements in $U$ are covered,
	i.e. all nodes in $U$ are neighbors of some node in $S$.

We can encode the Set Cover instance as an instance of SM-PCG as follows.
We use the same graph $G$ as the social graph for SM-PCG with edges oriented from
	nodes in $V$ to nodes in $U$, and the influence probability of
	all edges are $1$.
The target set is $U$, with coverage threshold $\eta = |U|$.
The probability threshold $P=1/2$ (actually since the diffusion in this setting is deterministic,
	any $P\in (0,1)$ works).
If the Set Cover instance has a solution, then every node in $U$ must be connected to some node
	in $V$.
In this case, for any solution $S$ for the above SM-PCG instance, if $S\cap U\ne \emptyset$,
	we can replace each node $u\in S\cap U$ with its neighbor node $v\in V$ so that we find
	a set $S'\subset V$ and it must be the case that $|S'|=|S|$ and $S'$ is also a solution to
	SM-PCG.
Then $S'$ must be a solution to the Set Cover instance.
Conversely, any solution to the Set Cover instance is also a solution to the SM-PCG instance.
Since Set Cover problem is NP hard, and cannot be approximated within a ratio of
	$(1-\varepsilon)\ln n$ unless NP has $n^{O(\log\log n)}$-time deterministic
	algorithms~\cite{feige98}, so does the SM-PCG problem.
\qedsymbol \end{proof}

With the above hardness result, we set our goal as to find
	algorithms that solve the SM-PCG problem with approximation ratio close to
	$\ln n$.

\section{Useful Results and Tools}\label{sec:tools}

In this section, we provide some useful results and tools in preparation for our algorithm design.

Almost all previous work on social influence maximization or seed minimization
	is based on submodular function optimization techniques.
Consider a set function $f(\cdot)$ which maps subsets of a finite ground set into real number set $\mathbb{R}$. We say that $f(\cdot)$ is \emph{submodular} if for any subsets $S\subseteq T$ and any element $u\not\in T$, $f(S\cup \{u\})-f(S)\geq f(T\cup \{u\})-f(T)$.
Moreover, we say that $f(\cdot)$ is \emph{monotone} if for any subsets $S\subseteq T$, $f(S)\leq f(T)$.

Consider a monotone and submodular function $f(\cdot)$ on subsets of nodes
	in the social graph $G=(V,E)$.
Suppose that each node $v\in V$ has a cost $c(v)$, given by a cost function $c:V\rightarrow \mathbb{R}^+$.
The cost of a subset $S$ is defined as $c(S) = \sum_{v\in S} c(v)$.
In~\cite{goyal2012minimizing}, Goyal et al. investigate the problem of finding a subset $S\subseteq V$
	with minimum cost such that $f(S)$ is at least some given threshold $\eta$.
As in many optimization tasks for submodular functions, the following greedy algorithm is applied
	to solve the problem: starting from the emptyset $S_0=\emptyset$, in the $i$-th iteration
	with $i=1,2,\ldots$, find a node $v_i$ that provides the largest marginal gain on $f$ per-unit
	cost, that is find
\[
v_i = \argmax_{v\in V\setminus S_{i-1}} \frac{f(S_{i-1}\cup \{v\})-f(S_{i-1})}{c(v)},
\]
and add $v_i$ to $S_{i-1}$ to obtain $S_i$; continue this process until iteration $j$
	in which $f(S_j)\ge \theta$, where $\theta$ is a threshold that could be $\eta$ or some other value chosen by the algorithm as the stopping criteria,
	and output $S_j$ as the selected subset $S$.
However, 	generally computing $f(\cdot)$ exactly is \#P-hard, but for most influence spread models, it can be estimated by Monte Carlo simulation as accurately as possible. We say an estimation $\hat{f}(\cdot)$ is
a \emph{$\gamma$-multiplicative error estimation} of $f(\cdot)$, if for any subset $S$, $|\hat{f}(S) - f(S)| \le \gamma f(S)$.	
Goyal et al. show the following bicriteria approximation result for the above greedy algorithm when $\gamma=0$.

%
\begin{theorem}
\cite{goyal2012minimizing} Let $G=(V,E)$ be a social graph, with cost function $c:V\rightarrow \mathbb{R}^+$ on the nodes of the graph.
Let $f(\cdot)$ be a nonnegative, monotone and submodular set function on the subsets of nodes.
Given a threshold $0<\eta \le f(V)$,
	let $S^*\subseteq V$ be a subset of minimum cost such that $f(S^*)\geq \eta$. Let $\varepsilon>0$ be any shortfall and let $S$ be the greedy solution satisfying
	$f(S)\ge \eta-\varepsilon$. Then, we have $c(S)\leq c(S^*)\cdot(1+\ln \frac{\eta}{\varepsilon})$.
When the costs on nodes
	are uniform, the approximation factor can be improved to $\lceil \ln\frac{\eta}{\varepsilon} \rceil$.
\label{bicriteria}
\end{theorem}

Based on their idea, for the case of uniform node cost and $\eta<f(V)$, we slightly improve their result by removing the bicriteria restriction and generalizing to the case of $\gamma\ge 0$.
	

\begin{theorem}
Let $G=(V,E)$ be a social graph, and
let $f(\cdot)$ be a nonnegative, monotone and submodular set function on the subsets $|V|$.
Given a threshold $0< \eta < f(V)$, 
	let $S^*\subseteq V$ be a subset of minimum size such that $f(S^*)\geq \eta$, and
	$S$ be the greedy solution using a $\gamma$-multiplicative error estimation function $\hat{f}(\cdot)$
	with the stopping criteria $\hat{f}(S)\ge (1+\gamma)\eta$.
For any $0\le \varepsilon_0 \le 1$, for any $0\le \gamma\le  \frac{\varepsilon_0(f(V)-\eta)}{8|V|(f(V)+\eta|V|)}$, we have
	$f(S)\ge \eta$, and
$|S| \le \alpha |S^*|+1$ where $\alpha=\max \{\left\lceil \ln \left(\frac{(1+\varepsilon_0)\eta |V|}{f(V)-\eta }\right)  \right\rceil, 0\}$.
\label{thm:greedy}
\end{theorem}


Note that when $\eta = \Theta(f(V))$, we have $\gamma\le \frac{\varepsilon_0(f(V)-\eta)}{8|V|(f(V)+\eta|V|)} = \Theta(\frac{\varepsilon_0}{|V|^2})$.

Let $S_i$ be the set containing the first $i$ seeds generated by the greedy algorithm with estimation $\hat{f}(\cdot)$. Let $\eta_i = \eta - f(S_i)$ and $\hat{\eta}_i = (1+\gamma)\eta - \hat{f}(S_i)$. Let $k$ be the size of $S^*$.
\begin{lemma}
For any $S\subset V$ with $f(S)<\eta$, there exists a node $x\in V\setminus S$ satisfying $f(S\cup \{x\}) - f(S) \ge \frac{\eta -f(S)}{k}$.
\label{le:diff}
\end{lemma}
\begin{proof}
Assume $\forall x\in V\setminus S, f(S\cup \{x\}) - f(S) < \frac{\eta-f(S)}{k}$. Let $S'=S^*\setminus S$.
\begin{eqnarray*}
f(S^* \cup S) &\le & f(S) +\sum_{x\in S'} \left( f(S\cup \{x\} - f(S) \right) \quad \mbox{(by submodularity of $f(\cdot)$)} \\
&< & f(S) + k\cdot \frac{\eta- f(S)}{k} \\
&= & \eta.
\end{eqnarray*}
It is a contradiction, since $f(S^*\cup S)\ge f(S^*) \ge \eta$. Thus, the lemma holds.
\qedsymbol
\end{proof}

\noindent {\it Proof of Theorem~\ref{thm:greedy}}.
Since for any seed set $S$, $|\hat{f}(S)- f(S)|\le \gamma f(S)$, we have
\begin{equation}
(1-\gamma)f(S) \le \hat{f}(S) \le (1+\gamma)f(S),
\label{eq:esterr1}
\end{equation}
and
\begin{equation}
\frac{1}{1+\gamma}\hat{f}(S) \le f(S) \le \frac{1}{1-\gamma}\hat{f}(S).
\label{eq:esterr2}
\end{equation}
For the output seed set $S$ of greedy algorithm with stopping criteria $\hat{f}(S)\ge (1+\gamma)\eta$, 
it is easy to see $f(S)\ge \frac{\hat{f}(S)}{1+\gamma} \ge \eta$. Thus, we mainly focus on proving inequality $|S| \le \alpha |S^*|+1$.

Let
\[ \varepsilon = \frac{f(V)}{1+\gamma}\left(\frac{1}{|V|} -2\gamma \right) -\frac{\eta}{(1-\gamma)|V|} . \]
If $\varepsilon \ge \eta$, we claim that $S$ only contains one seed. By a similar analysis of Lemma~\ref{le:diff}, there exists a node $x\in V$ satisfying $f(\{x\}) \ge \frac{f(V)}{|V|}$. Then,
\begin{eqnarray*}
\hat{f}(\{x\}) &\ge & (1-\gamma)f(\{x\}) \\
&\ge & \frac{(1-\gamma)f(V)}{|V|} \\
&\ge & f(V)\left( \frac{1}{|V|} -2\gamma \right) \\
&\ge & (1+\gamma)\varepsilon \\
&\ge & (1+\gamma) \eta.
\end{eqnarray*}
Thus, $\hat{f}(S)\ge (1+\gamma)\eta$ and $|S| \le \alpha |S^*|+1$. In the following, we prove the case of $\varepsilon < \eta$.

We first consider $S_l$ satisfying $\hat{f}(S_l)\ge (1+\gamma)(\eta - \varepsilon)$ and $\hat{f}(S_{l-1})< (1+\gamma)(\eta - \varepsilon)$.
We compute the difference between $\hat{f}(S_i)$ and $\hat{f}(S_{i-1})$, for all $1\le i\le l-1$. Suppose $S_0=\emptyset$.
Since
\begin{eqnarray*}
f(S_{i-1}) &\le & \frac{\hat{f}(S_{i-1})}{1-\gamma} \\
&<& \frac{1+\gamma}{1-\gamma}(\eta-\varepsilon) \\
&=& \eta - \frac{1}{1-\gamma}((1+\gamma)\varepsilon - 2\gamma\eta) \\
&\le & \eta-(\varepsilon - 2\gamma\eta) \\
&= & \eta-\left(\frac{f(V)}{1+\gamma}\left(\frac{1}{|V|} -2\gamma \right) -\frac{\eta}{(1-\gamma)|V|} - 2\gamma\eta\right) \\
&\le &\eta - \left( \frac{f(V)-\eta}{|V|} -\gamma\left( \frac{f(V)+\eta}{|V|} + 2f(V) + 2\eta \right) \right) \quad \mbox{(by $\gamma\ge 0$)} \\
&\le &\eta - \left( \frac{f(V)-\eta}{|V|} -4\gamma\left( f(V) + \eta \right) \right) \\
&\le & \eta -  \frac{f(V)-\eta}{2|V|} \quad \mbox{(by $\gamma\le \frac{\varepsilon_0(f(V)-\eta)}{8|V|(f(V)+\eta|V|)}$)} \\
&\le & \eta.
\end{eqnarray*}
By Lemma~\ref{le:diff}, there exists a node $x\in V\setminus S_{i-1}$ satisfying $f(S_{i-1}\cup \{x\}) -f(S_{i-1})\ge \frac{\eta_{i-1}}{k}$.

\begin{eqnarray*}
&& \hat{f}(S_{i}) -\hat{f}(S_{i-1}) \\
&\ge & \hat{f}(S_{i-1}\cup \{x\}) -\hat{f}(S_{i-1}) \\
&\ge & (1-\gamma)f(S_{i-1}\cup \{x\}) - (1+\gamma)f(S_{i-1}) \quad \mbox{(by~\eqref{eq:esterr1})}\\
&= & f(S_{i-1}\cup \{x\})-f(S_{i-1}) -\gamma(f(S_{i-1}\cup \{x\})+f(S_{i-1})) \\
&\ge & \frac{\eta-f(S_{i-1})}{k} - \frac{2\gamma(1+\gamma)}{1-\gamma}\eta \quad (\mbox{by Lemma~\ref{le:diff} and the fact $f(S_{i-1}\cup \{x\}), f(S_{i-1})\le \frac{1+\gamma}{1-\gamma}\eta$}) \\
&\ge &\frac{\eta - \hat{f}(S_{i-1})/(1-\gamma)}{k} - \frac{2\gamma(1+\gamma)}{1-\gamma}\eta \quad \mbox{(by~\eqref{eq:esterr2})}\\
&= &\frac{\hat{\eta}_{i-1}}{(1-\gamma)k} - \frac{2\gamma\eta}{1-\gamma}\left(\frac{1}{k} +\gamma +1\right).
\end{eqnarray*}
By the definition of $\hat{\eta}_i$, we have
\[ \hat{\eta}_{i-1} -\hat{\eta}_i  = \hat{f}(S_{i}) -\hat{f}(S_{i-1}) \ge \frac{\hat{\eta}_{i-1}}{(1-\gamma)k} - \frac{2\gamma\eta}{1-\gamma}\left(\frac{1}{k} +\gamma +1\right), \]
that is,
\[ \hat{\eta}_i\le \left(1-\frac{1}{(1-\gamma)k}\right) \hat{\eta}_{i-1} +  \frac{2\gamma\eta}{1-\gamma}\left(\frac{1}{k} +\gamma +1\right). \]
Since $\hat{f}(S_{l-1})< (1+\gamma)(\eta-\varepsilon)$, thus $\hat{\eta}_{l-1}> (1+\gamma)\varepsilon$. Let $a=1-\frac{1}{(1-\gamma)k}$ and $b=\frac{2\gamma\eta}{1-\gamma}\left(\frac{1}{k} +\gamma +1\right)$.
\begin{eqnarray*}
\hat{\eta}_{l-1} &\le & a\hat{\eta}_{l-2} + b \\
&\le & a^{l-1} (1+\gamma)\eta + \frac{1-a^{l-1}}{1-a}b \\
&\le & a^{l-1} (1+\gamma)\eta + \frac{1}{1-a}b.
\end{eqnarray*}
Thus,
\[ (1+\gamma)\varepsilon < \left(1-\frac{1}{(1-\gamma)k} \right)^{l-1}(1+\gamma)\eta + 2\gamma\eta((1+\gamma)k +1). \]
Since $\forall z, 1+z\le e^z$,
\[ (1+\gamma)\varepsilon < e^{-\frac{l-1}{(1-\gamma)k}}(1+\gamma)\eta + 2\gamma\eta((1+\gamma)k +1). \]
It means
\[ l < (1-\gamma)k \ln \left(\frac{\eta}{\varepsilon- 2\gamma\eta(k+\frac{1}{1+\gamma})}\right)  +1. \]
Since $l$ is an integer,
\[ l \le \left\lceil (1-\gamma)k \ln \left(\frac{\eta}{\varepsilon- 2\gamma\eta(k+\frac{1}{1+\gamma})}\right) \right\rceil. \]
Since $k\le |V|$,
\[ |S_l| \le \left\lceil (1-\gamma) \ln \left(\frac{\eta}{\varepsilon- 2\gamma\eta(|V|+\frac{1}{1+\gamma})}\right) \right\rceil |S^*|. \]

If $\hat{f}(S_l)\ge (1+\gamma)\eta$, then let $S=S_l$, we have done. Otherwise, $\hat{f}(S_l) < (1+\gamma)\eta$.

By (2), we know that $f(S_l)\le \frac{\hat{f}(S_l)}{1-\gamma} < \frac{1+\gamma}{1-\gamma}\eta$. By a similar analysis of Lemma~\ref{le:diff}, we know that there exists a node $x\in V\setminus S_l$ satisfying
\[f(S_l\cup \{x\})-f(S_l) \ge \frac{f(V)-\frac{1+\gamma}{1-\gamma}\eta}{|V|}. \]
We consider the marginal increment of $x$ on $\hat{f}(S_l)$.
\begin{eqnarray*}
&& \hat{f}(S_l\cup \{x\}) -\hat{f}(S_l) \\
&\ge & (1-\gamma)f(S_l\cup \{x\}) -(1+\gamma)f(S_l) \quad \mbox{(by~\eqref{eq:esterr1})}\\
&= & f(S_l\cup\{x\}) - f(S_l) -\gamma \left(f(S_l\cup \{x\}) +f(S_l) \right)  \\
&\ge & \frac{f(V)-\frac{1+\gamma}{1-\gamma}\eta}{|V|} -2\gamma f(V)\\
&= &  f(V)\left(\frac{1}{|V|} -2\gamma\right) -\frac{(1+\gamma)\eta}{(1-\gamma)|V|} \\
&= & (1+\gamma)\varepsilon.
\end{eqnarray*}
Thus, $\hat{f}(S_{l+1})\ge (1+\gamma)\eta$. Let $S=S_{l+1}$, we have
\[ |S| \le \left\lceil (1-\gamma)\ln\left( \frac{\eta}{ \frac{f(V)}{1+\gamma}\left(\frac{1}{|V|}-2\gamma\right) -\frac{\eta}{(1-\gamma)|V|} -2\gamma\eta\left(|V|+ \frac{1}{1+\gamma}\right) } \right)\right\rceil |S^*| +1. \]
Since $\gamma\le \frac{\varepsilon_0(f(V)-\eta)}{8|V|(f(V)+\eta|V|)}$, we have
\begin{eqnarray*}
&& \frac{\eta}{ \frac{f(V)}{1+\gamma}\left(\frac{1}{|V|}-2\gamma\right) -\frac{\eta}{(1-\gamma)|V|} -2\gamma\eta\left(|V|+ \frac{1}{1+\gamma}\right) } \\
&= & \frac{\eta |V|}{\frac{f(V)}{1+\gamma} - \frac{\eta}{1-\gamma} - 2\gamma \left(\frac{f(V) |V|}{1+\gamma} + \eta |V|^2 +\frac{\eta |V|}{ 1+\gamma}\right) } \\
&\le & \frac{\eta |V|}{f(V)-\eta - \gamma\left( f(V)+\eta + 2(1-\gamma)f(V)|V| +2(1-\gamma^2)\eta |V|^2 + 2(1-\gamma)\eta |V|\right) } \\
&\le & \frac{\eta |V|}{f(V)-\eta - \gamma\left( f(V)+\eta + 2f(V)|V| +2\eta |V|^2 + 2\eta |V|\right) } \\
&\le & \frac{\eta |V|}{f(V)-\eta - 4\gamma\left(f(V)|V| +\eta |V|^2\right) } \\
&\le & \frac{\eta |V|}{(1-\frac{\varepsilon_0}{2})(f(V)-\eta )} \quad \mbox{(by $\gamma\le \frac{\varepsilon_0(f(V)-\eta)}{8|V|(f(V)+\eta|V|)}$)}\\
&\le & \frac{(1+\varepsilon_0)\eta |V|}{f(V)-\eta} \quad \mbox{(by $\varepsilon_0\le 1$)} .
\end{eqnarray*}
By the fact $\varepsilon<\eta$, we know that $\left\lceil \ln \left(\frac{(1+\varepsilon_0)\eta |V|}{f(V)-\eta }\right)  \right\rceil \ge \left\lceil \ln  \frac{\eta}{\varepsilon}  \right\rceil > 0$. Thus,
\[ |S| \le \left\lceil \ln \left(\frac{(1+\varepsilon_0)\eta |V|}{f(V)-\eta }\right)  \right\rceil |S^*| +1. \]
The theorem holds.
\qedsymbol

Kempe et al. show that set function $E[\Inf(S)]$ for {\em expected influence
	coverage} is monotone and
	submodular under the IC model~\cite{kempe2003maximizing}.
Therefore, if our problem is to find a seed set of minimum size such that the expected
	influence coverage
	is at least a threshold value $\eta$,
	Theorem~\ref{thm:greedy} already provides the approximation guarantee of the greedy algorithm.
We call this problem the {\em seed minimization with expected coverage guarantee
	(SM-ECG)}, to differentiate with the problem concerned in this paper ---
	seed minimization with {\em probabilistic} coverage guarantee (SM-PCG).

For the SM-PCG problem, we want the influence coverage to be
	at least $\eta$ with a guaranteed probability $P$.
This seemingly minor change from SM-ECG actually alters the nature of the problem.
The SM-PCG corresponds to two variants of set functions, but neither of them is submodular.
In the first variant, we fix influence threshold $\eta$, and define $f_{\eta}: 2^{|V|}\rightarrow \mathbb{R}^+$ where $f_{\eta}(S) = \Pr(\Inf(S)\ge \eta)$.
In the second variant, we fix probability $P$, and define $g_P: 2^{|V|}\rightarrow \mathbb{R}^+$ where $g_P(S)=\max_{\eta':\Pr(\Inf(S)\geq \eta')\geq P} \eta'$.
Neither $f_\eta(\cdot)$ nor $g_P(\cdot)$ is submodular, as shown by the two examples below.
For $f_{\eta}$, see Figure \ref{nonsubmodular}, $G$ is a bipartite graph where all edges are associated with probability $1$, and $U$ contains all the nodes in the lower part. We fix $\eta=5$.
Let $S=\{a\}$ and $T=\{a,b\}$, then $f_{\eta}(S\cup \{u\})-f_{\eta}(S)=0$, since neither
	$S$ nor $S\cup \{u\}$ could reach $5$ nodes in $U$.
Similarly, $f_{\eta}(T)=0$.
However, $f_{\eta}(T\cup \{u\}) = 1$, since $5$ nodes are reached by $T\cup \{u\}$.
Therefore, $f_{\eta}(T\cup \{u\})-f_{\eta}(T) >f_{\eta}(S\cup \{u\})-f_{\eta}(S)$, and
	thus $f_\eta(\cdot)$ is not submodular.
For $g_P$, see Figure \ref{nonsubmodular2}, $G$ is a bipartite graph where all edges are associated with probability $0.5$, and $U=\{u\}$. We set $P=0.8$. Let $S=\{a\}$ and $T=\{a, b\}$, then $g_P(S\cup \{c\})-g_P(S)=0$ and $g_P(T\cup \{c\})-g_P(T)=1$. Since $g_P(S\cup \{c\})-g_P(S)< g_P(T\cup \{c\})-g_P(T)$, $g_P$ is not submodular.

\begin{figure}[t]
\begin{center}
\includegraphics[scale=0.5]{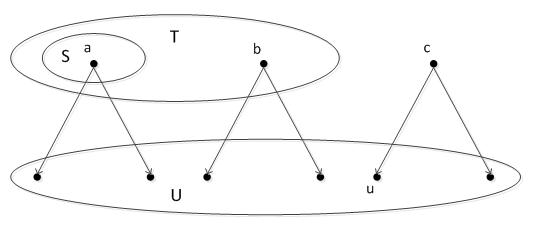}
\caption{Function $f_{\eta}$ is nonsubmodular}
\label{nonsubmodular}
\end{center}
\end{figure}
\begin{figure}[t]
\begin{center}
\includegraphics[scale=0.5]{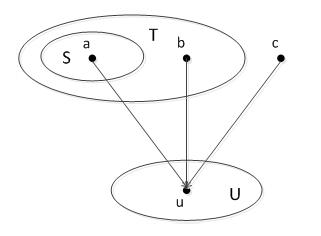}
\caption{Function $g_P$ is nonsubmodular}
\label{nonsubmodular2}
\end{center}
\end{figure}

Since neither $f_\eta(\cdot)$ nor $g_P(\cdot)$ is submodular, we cannot apply
	Theorem~\ref{thm:greedy} on $f_\eta(\cdot)$ or $g_P(\cdot)$ to solve the SM-PCG problem.
In this paper, we address this non-submodular optimization problem by relating
	it to the SM-ECG problem through a concentration assumption on random variable
	$\Inf(S)$ for certain seed sets $S$.	
We will use the following concentration inequalities in the later sections.

\begin{fact}[Chebyshev's inequality] \label{fact:chebyshev}
Let $X$ be a random variable with finite expectation $E[X]$ and finite variance $\Var(X)$. Then for any real value $t>0$,
$$\Pr(|X-E[X]|\geq t)\leq \frac{\Var(X)}{t^2}.$$
\end{fact}

\begin{fact}[Hoeffding's inequality] \label{fact:hoeffding}
Let $X_1,\ldots,X_n$ be independent random variables.
Assume that the $X_i$ are almost surely bounded, that is, assume for $1\leq i\leq n$ that
$\Pr(X_i\in[a_i,b_i])=1$.
We define the sum of these variables
$X=X_1+\cdots+X_n.$
Then, for any constant $t>0$,
$$\Pr(X-E[X]\geq t)\leq \exp\left(-\frac{2t^2}{\sum_{i=1}^n (b_i-a_i)^2}\right),$$
$$\Pr(|X-E[X]|\geq t)\leq 2 \exp\left(-\frac{2t^2}{\sum_{i=1}^n (b_i-a_i)^2}\right).$$
\end{fact}

\section{Influence Coverage Computation}
\label{sec:infcomp}

Before working on the SM-PCG problem directly, we first address the related computation issue
	when a seed set $S$ is given.
As we mentioned in last section, there are two variants in influence coverage computation.
The first variant is that, given a seed set $S$ and a coverage threshold $\eta$, we need to compute
	the probability $f_\eta(S)$ that $S$ can activate at least $\eta$ nodes in $U$.
Note that we have $E[\Inf(S)] = \sum_{i=1}^{n-1} (f_{i}(S)-f_{i+1}(S))\cdot i + f_n(S)\cdot n$.
Thus the exact computation of $f_\eta(S)$ must be \#P-hard in the IC model since computing
	expected influence coverage $E[\Inf(S)]$ of seed set $S$ has shown to be \#P-hard in the IC
	model~\cite{chen2010scalable}.
However, we can use Monte Carlo simulations to compute an accurate estimate of the probability.
Algorithm \ref{computeprobability} shows the procedure $\MCCP[R]$ for this task, which
	simulate the diffusion from seed set $S$ for $R$ runs and use the fraction of runs in which
	the number of active nodes in $U$ reaches $\eta$ as the estimate of the probability.

\begin{algorithm}[thb]
\renewcommand{\algorithmicrequire}{\textbf{Input:}}
\renewcommand\algorithmicensure {\textbf{Output:}}
    \caption{Function $\MCCP[R]$: $R$ is a tuning parameter controlling the accuracy of the estimate} \label{alg:compProb}
    \begin{algorithmic}[1]
    \REQUIRE ~~ $G=(V,E), \{p_{u,v}\}_{(u,v)\in E},  U, S, \eta$\\
    \ENSURE ~~  estimate of $P=\Pr(\Inf(S)\geq \eta)$\\
\STATE $t=0$
\FOR {$i=1$ to $R$}
    \STATE simulate IC diffusion with seed set $S$
    \STATE $N_i =$ number of final active nodes in $U$
    \IF {$N_i\geq \eta$}
        \STATE \ $t=t+1$
    \ENDIF
\ENDFOR
\RETURN $t/R$

\end{algorithmic}
\label{computeprobability}
\end{algorithm}

	
The following lemma shows the relationship between the number of simulations $R$ and
	the accuracy of the estimate.
\begin{lemma}
Let $\hat{P}$ be the estimate of true value $P = $ $\Pr(\Inf(S)$ $ \geq \eta)$ output by $\MCCP[R]$
	in Algorithm~\ref{alg:compProb}.
To guarantee an error of at most $\varepsilon$, i.e. $|\hat{P}-P| \le \varepsilon$ \footnote{This lemma holds when $\varepsilon > P$. However, we usually set $\varepsilon$ smaller than $P$ to make the estimate more reasonable.}, with
	probability at least $1- 1/n^\delta$, it is sufficient to set $R \ge \ln(2n^\delta)/(2\varepsilon^2)$.
\label{lem:compProb}
\end{lemma}
\begin{proof}
Let $X_i$ be a boolean random variable, with $X_i=1$ meaning the influence coverage of the $i$-th
	simulation run in the algorithm $\MCCP[R]$ is at least $\eta$, and $0$ otherwise.
Let $X = \sum_{i=1}^R X_i$.
Then we have $X = \hat{P}\cdot R$, and $E[X] = P\cdot R$.
Thus we can apply Hoeffding's Inequality as given in Fact~\ref{fact:hoeffding} and obtain
\begin{align*}
\Pr(|\hat{P}-P| \ge \varepsilon) & = \Pr(|X - E[X]| \ge R\varepsilon) \\
& \le 2 \exp(-2R\varepsilon^2) \le \frac{1}{n^\delta},
\end{align*}
where the last inequality uses condition $R \ge \ln(2n^\delta)/(2\varepsilon^2)$.
\qedsymbol \end{proof}
The second variant is that, given a seed set $S$ and a specified probability $P$,
	we need to compute the maximum influence coverage $\eta$ of $S$ with at least probability $P$,
	that is, $\eta = \max_{\eta': \Pr(\Inf(S)\ge \eta') \ge P} \eta'$.
Unlike the first variant, we show below that this problem is \#P-hard to approximate to
	any non-trivial ratio.
We say that an algorithm approximates a true value $v$ for a computing problem with ratio $\alpha>1$ if the output of algorithm
	$\hat{v}$ satisfies $v/\alpha \le \hat{v}\le \alpha v$.
Note that if the range of value $v$ is from $1$ to $n$, then using $\hat{v}=n^{1/2}$ gives a trivial
	approximation ratio of $\alpha = n^{1/2}$.
\begin{theorem}
For any fixed probability $P\in (0,1)$, the problem of computing
	$\eta = \max_{\eta': \Pr(\Inf(S)\ge \eta') \ge P} \eta'$ given
	a directed social graph $G=(V, E)$, influence probabilities
	$\{p_{u,v}\,|\, (u,v)\in E\}$, target set $U=V$, and a seed set $S$ is \#P-hard to approximate
	within a ratio of $|V|^{1/2-\varepsilon}$ for any $\varepsilon>0$.
\end{theorem}
Note that we treat $P$ as a fixed parameter of the problem rather than as part of the input to
	the computation problem, which makes the result stronger.
\begin{proof}
We prove the theorem by a reduction from the \#P-complete
	counting problem of $s$-$t$ connectivity in a directed graph \cite{valiant1979complexity}. Given a directed graph, two specified nodes $s$ and $t$, the objective of this problem is to find the total number of subgraphs with the same set of nodes but a subset of
	edges in which there exists at least one path from $s$ to $t$.
This problem is equivalent to the following problem:
given a directed graph and two different nodes $s$ and $t$, each edge in that graph has an independent probability of
	$1/2$ to appear or disappear, and the objective is to compute the probability that $s$
	reaches $t$.
	
We first reduce the above $s$-$t$ connectivity problem to
	its decision version, that is, given a graph $G=(V,E)$, two
	nodes $s$ and $t$, a probability $Q$, and each edge having an independent probability of $1/2$
	to appear or disappear, ask whether the probability of $s$ reaching $t$ is at least $Q$ or not.
If this decision version is solvable, we can do a binary search to find the actual probability of $s$
	reaching $t$.
Note that since each edge has probability $1/2$ to appear or disappear, the probability of
	$s$ reaching $t$ is a multiple of $1/2^{|E|}$, which means we can find its exact value using
	$|E|$ queries to the decision version of the problem.
Therefore, the decision version of the $s$-$t$ connectivity problem is \#P-hard.
	
We now reduce the decision version of the $s$-$t$ connectivity problem to the problem of computing
	$\eta = \max_{\eta': \Pr(\Inf(S)\ge \eta') \ge P} \eta'$.
Consider any instance of the decision version of the $s$-$t$ connectivity problem with
	graph $G=(V,E)$ and probability $Q$.
If $Q=0$ or $1$, the decision problem is trivial, and thus we assume $0< Q<1$.
Let $n = |V|$.
We construct a new graph $G'=(V',E')$ from $G$ in the following way, as shown in Figure~\ref{s-t}.
We add nodes $u$ together with  $N$ additional nodes to $G$,  where $N=n^c$ and $c$ is a constant to be determined shortly.
We also add directed edges $(s,u)$ and $(t,u)$,  and directed edges from $u$ to all the
	$n^c$ additional nodes.
The influence probabilities of all original edges in $G$ are $1/2$, while
	the influence probabilities of all new edges except $(s,u)$ and $(t,u)$ are $1$.
If $Q\ge P$, $p_{t,u} = P/Q$ and $p_{s,u}=0$; if $Q < P$,
	$p_{t,u} = 1$ and $p_{s,u}=1 - (1-P)/(1-Q)$.\footnote{We require that algorithm $Alg$
	could handle any rational number input on influence
	probabilities. Then the constructed $p_{t,u}$ and $p_{s,u}$ can be encoded as a rational
	number with lengths polynomial to the lengths of numbers $P$ and $Q$.}
Let $p_0$ be the probability of $s$ reaching $t$ in the original graph $G$, which is the same
	as the probability of $s$ reaching $t$ in $G'$.
Let $p_1$ be the probability of $s$ reaching $u$ in $G'$.
It is easy to check that with the above setup, $p_0\ge Q$ if and only if $p_1 \ge P$.

\begin{figure}[h]
\centering
\includegraphics[scale=0.5]{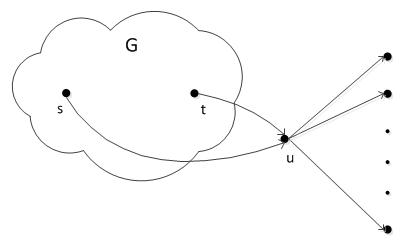}
\caption{$G'$: influence probabilities of all original edges in $G$ are 1/2, influence probabilities of all new added edges except $(s,u)$ and $(t,u)$ are 1, and influence probabilities of $(s,u)$ and $(t,u)$ depend on the value of $Q$ and $P$}
\label{s-t}
\end{figure}

Let $S=\{s\}$.
Assume that there exists an approximation algorithm $Alg$ that outputs $\hat{\eta}$, which
	approximates the true value
	$\eta = \max_{\eta': \Pr(\Inf(S)\ge \eta') \ge P} \eta'$
	with an approximation ratio $|V|^{1/2-\varepsilon}$ for some $\varepsilon>0$,
	where $V$ is the node set of the input graph to $Alg$.
We choose a sufficiently large constant $c$ such that $N = n^c > n |V'|^{1-2\varepsilon}$ where
	$|V'| = N + n +1$ is the total number of nodes in graph $G'$.
Then, for constructed $G'$ with the chosen constant $c$,
	we can use $Alg$ to distinguish whether $p_1 \ge P$ as follows.
If $p_1 \ge P$, then $s$ reaches $u$ and thus all $N$ additional nodes with probability at least $P$,
	which means the true value $\eta$ is at least $N$.
Thus we have $\hat{\eta} \ge N / |V'|^{1/2-\varepsilon}$.
If $p_1 < P$, the true value $\eta$ is at most $n$, and thus
	$\hat{\eta} \le n |V'|^{1/2-\varepsilon}$.
By our choice of $c$, we know that $N / |V'|^{1/2-\varepsilon} > n |V'|^{1/2-\varepsilon}$.
Therefore, we can use the condition $\hat{\eta} \ge N / |V'|^{1/2-\varepsilon}$ to determine if
	$p_1 \ge P$.
Since our construction guarantees that $p_1 \ge P$ if and only if $p_0 \ge Q$, we can use
	the condition $\hat{\eta} \ge N / |V'|^{1/2-\varepsilon}$ to answer the decision question
	of the $s$-$t$ connectivity problem.
This implies that our problem of computing
	$\eta = \max_{\eta': \Pr(\Inf(S)\ge \eta') \ge P} \eta'$
	within the ratio of $|V|^{1/2-\varepsilon}$ for any $\varepsilon > 0$ is \#P-hard.
\qedsymbol \end{proof}

\section{Approximation Algorithm}
\label{sec:algo}

In this section, we overcome the nonsubmodularity nature of the SM-PCG problem
	discussed in Section~\ref{sec:tools} by connecting it with the submodular problem SM-ECG.
We first provide the general algorithm, and then show that the algorithm returns a seed set
	that approximates the optimal solution with both a multiplicative ratio and an additive
	error.
The multiplicative ratio is due to the connection with the SM-ECG problem.
For the additive error term, we show that it would be nontrivial when certain
	concentration assumption on influence coverages holds.

\begin{algorithm}[thb]
\renewcommand{\algorithmicrequire}{\textbf{Input:}}
\renewcommand\algorithmicensure {\textbf{Output:}}
    \caption{$\MinSeed[\varepsilon]$: $\varepsilon \in [0, (1-P)/2)$ is a control parameter}
    \begin{algorithmic}[1]
    \REQUIRE $G=(V,E), \{p_{u,v}\}_{(u,v)\in E}, U, \eta,P$\\
    \ENSURE seed set $S$, which is an approximation to $S^*=\argmin_{S': \Pr(\Inf(S')\geq \eta)\geq P}\{|S'|\}$\\
\STATE \label{line:greedyb}
	$S_0 = \emptyset$\\
\FOR {$i=1$ to $n$}
    \STATE select $u=\argmax_v \{\hat{E}[\Inf(S_{i-1}\cup \{v\})]-\hat{E}[\Inf(S_{i-1})]\}$ \\
    \STATE $S_i=S_{i-1}\cup \{u\}$ \\
    \STATE \label{line:getprob}
    	$\prob = \ComputeProbability(G, \{p_{u,v}\}_{(u,v)\in E}, U, \eta, S_{i})$
    \IF{$\prob \ge P+\varepsilon$}
    	\RETURN $S_i$
    \ENDIF
\ENDFOR \label{line:greedye}
%
\end{algorithmic}
\label{general-alg}
\end{algorithm}

Algorithm \ref{general-alg} illustrates algorithm $\MinSeed$ for solving the SM-PCG problem.
The algorithm builds up a sequence of subsets $S_0, S_1,S_2, \ldots$, where for any $i \ge 1$,
	$S_i$ contains one more element
	$u$ than $S_{i-1}$ such that $u$ provides the largest marginal increase in {\em expected} influence
	coverage to seed set $S_{i-1}$.
The way of constructing seed sets $S_i$'s is in line with the greedy approach as discussed
	in Section~\ref{sec:tools}.
In our algorithm, $\hat{E}[\Inf(\cdot)]$ is a $\gamma$-multiplicative error estimation of exact expected influence $E[\Inf(\cdot)]$.
Every time a new set $S_i$ is constructed, we compute the probability that
	the influence coverage of $S_i$ is at least $\eta$ (line~\ref{line:getprob}).
The $\ComputeProbability$ in line~\ref{line:getprob} is a generic function computing
	$\Pr(\Inf(S_i)\ge \eta)$, which could be
	$\MCCP[R]$ in Algorithm~\ref{alg:compProb} for general graphs, or $\BiCP$ in Algorithm~\ref{alg:BiCP} for one-way bipartite graphs, or
	some other functions for this purpose.
If the probability computed is at least $P+\varepsilon$, where $\varepsilon\in [0, (1-P)/2)$ is a parameter of
	the algorithm, we stop and return $S_i$ as the seed set found by the algorithm.
Parameter $\varepsilon$ is related to the accuracy of the function $\ComputeProbability$.
If $\ComputeProbability$ accurately computes $\Pr(\Inf(S)\ge \eta)$ (e.g. $\BiCP$ for
	one-way bipartite graphs), we set $\varepsilon$ to $0$.
If $\ComputeProbability$ only provides an estimate (e.g. $\MCCP[R]$ for general graphs),
	we set $\varepsilon$ to be an appropriate value related to the error term of the estimate given
	by the function.
We will discuss parameter $\varepsilon$ with more technical details later.

%

Let $S^*$ be the optimal seed set for the SM-PCG problem, that is,
	$S^* = \argmin_{S:\Pr(\Inf(S)\geq \eta )\geq P} |S|$.
Let $n=|V|$ and $m = |U|$.
Let ${\cal S} = \{S_1,S_2,\ldots, S_n=V\}$ be the sequence of greedy seed sets computed by algorithm
	$\MinSeed[\varepsilon]$ (considering the entire sequence even when $\MinSeed[\varepsilon]$
	actually stops).
Let $S_a$ be the output of $\MinSeed[\varepsilon]$ and $a$ is its index in
	sequence $\cal S$, and thus $S_{a-1}$ is the set in $\mathcal{S}$ just before $S_a$.
	
We define $c=\max\{\eta-E[\Inf(S^*)],0\}$ and $c'=\max\{E[\Inf(S_{a-1})]-\eta,0\}$.
Intuitively, we know that $\Pr(\Inf(S^*)\geq \eta )\geq P$, and $c$ indicates how much
	$E[\Inf(S^*)]$ could be smaller than $\eta$.
If $\Inf(S^*)$ concentrates well, $c$ should be small.
Similarly, we also know that $\Pr(\Inf(S_{a-1}) \ge \eta) < P+\varepsilon$, since $S_a$ is the
	first set satisfying $\Pr(\Inf(S_{a}) \ge \eta) \ge P+\varepsilon$.
Thus, $c'$ indicates how much $E[\Inf(S_{a-1})]$ could be larger than $\eta$, and if
	$\Inf(S_{a-1})$ concentrates well, $c'$ should be small.

The following theorem shows that the output $S_a$ of $\MinSeed[\varepsilon]$ approximates
	the optimal solution $S^*$ with $c$ and $c'$ included in the additive error term.

\begin{theorem} \label{thm:generalapprox}
For any $0\le \varepsilon_0\le 1$ and any $0\le \gamma\le \frac{\varepsilon_0(m-(\eta+c'))^2}{8mn(m+\eta n)}$. If $\hat{E}[\Inf(\cdot)]$ is a $\gamma$-multiplicative error estimation of $E[\Inf(\cdot)]$ for any subset of nodes, the size of the output by algorithm $\MinSeed[\varepsilon]$ approximates the size of the optimal
	solution in the following form:
\begin{equation}
|S_a|\le \left\lceil \ln \left(\frac{(1+\varepsilon_0)\eta n}{m-\eta}\right) \right\rceil |S^*| + \frac{(c+c')n}{m-(\eta+c')} +3 +\varepsilon_0.
\label{eq:approx}
\end{equation}
\label{thm:approxratio}
\end{theorem}

First, note that we assume $m > \eta$, so the multiplicative term above is well defined.
Moreover, $\eta + c'$ must be less than $m$, because otherwise $E[\Inf(S_{a-1})] = m = |U|$,
	which implies $\Pr(\Inf(S_{a-1}) = m) = 1$, contradicting the fact that
	$\Pr(\Inf(S_{a-1}) \ge \eta) < P+\varepsilon<1$.
Second, for the multiplicative ratio of  $\lceil \ln \frac{(1+\varepsilon_0)\eta n}{m-\eta} \rceil$, when $\eta$ is a constant fraction of $m$, i.e. $\eta = \beta m$
	where $\beta$ is a constant independent of $m$ and $n$,
	it is $\ln n + O(1)$, which is
	tight, since Theorem~\ref{thm:hard} already states that the ratio cannot be better
	than $\ln n$.
The additive error term involves $c$ and $c'$, and we will
	discuss it in more detail after providing
	the proof to the theorem below.
Third, when $\eta+c'$ is a constant fraction of $m$, $\gamma\le \frac{\varepsilon_0(m-(\eta+c'))^2}{8mn(m+\eta n)}= \Theta(\frac{\varepsilon_0}{n^2})$.
Fourth, by Chernoff bound, to achieve a $\gamma$-multiplicative error estimation of expected influence with probability $1-1/n$ for all subsets computed in our algorithm, it is sufficient to sample $\Theta(\gamma^{-2}n\log n)$ number of graphs for each set.

\begin{proof}
Let $i$ be the minimum index such that $S_i\in \mathcal{S}$ and $\hat{E}[\Inf(S_i)]\geq (1+\gamma)(\eta-c)$ (implying $E[\Inf(S_i)]\ge (\eta-c)$), and $S_i^*$ be the minimum-sized seed set such that $E[\Inf(S_i^*)]\geq \eta-c$.
Since $\gamma \le \frac{\varepsilon_0(m-(\eta+c'))^2}{8mn(m+\eta n)} \le \frac{\varepsilon_0(m-\eta)}{8n(m+\eta n)}$ and $\left\lceil \ln\left(\frac{(1+\varepsilon_0)(\eta-c) n}{m-(\eta-c)}\right) \right\rceil>0$, by
Theorem \ref{thm:greedy}, we have that
\[ |S_i| \le \left\lceil \ln\left(\frac{(1+\varepsilon_0)(\eta-c) n}{m-(\eta-c)}\right) \right\rceil |S^*_i| +1 \leq \left\lceil \ln\left(\frac{(1+\varepsilon_0)\eta n}{m-\eta}\right) \right\rceil |S^*_i| +1.\]
Since $E[\Inf(S^*)]\geq \eta-c$, we know that $|S_i^*|\leq |S^*|$,
\[ |S_i|\leq \left\lceil \ln\left(\frac{(1+\varepsilon_0)\eta n}{m-\eta}\right) \right\rceil |S^*| +1.\]

Let $j$ be the minimum index such that $S_j\in \mathcal{S}$ and $E[\Inf(S_j)]\geq \eta+c'$. Since $E[\Inf(S_{a-1})]\leq \eta+c'$, we know that $|S_j|\geq |S_{a-1}|$. To bound the difference between $|S_{a-1}|$ and $|S_i|$, it is sufficient to compute the difference between $|S_j|$ and $|S_i|$.

By the definition of $j$, we have that $E[\Inf(S_{j-1})]< \eta+c'$. Since $E[\Inf(S_i)]\ge \hat{E}[\Inf(S_i)]/(1+\gamma) \ge \eta-c$, we have $E[\Inf(S_{j-1})]-E[\Inf(S_i)]< c+c'$. For any $i<t<j$, by a similar analysis in Lemma~\ref{le:diff}, we know that there exists a node $x\in V\setminus S_t$ satisfying $E[\Inf(S_{t-1}\cup \{x\})]-E[\Inf(S_{t-1})] \ge \frac{m-(\eta+c')}{n}$. Then,
\begin{eqnarray*}
&& \hat{E}[\Inf(S_t)] -\hat{E}[\Inf(S_{t-1})] \\
&\ge & \hat{E}[\Inf(S_{t-1} \cup \{x\})] - \hat{E}[\Inf(S_{t-1})] \\
&\ge & (1-\gamma)E[\Inf(S_{t-1}\cup \{x\})] - (1+\gamma)E[\Inf(S_{t-1})] \\
&= & E[\Inf(S_{t-1} \cup \{x\})] -E[\Inf(S_{t-1})] -\gamma\left(  E[\Inf(S_{t-1} \cup \{x\})] +E[\Inf(S_{t-1})] \right) \\
&\ge & \frac{m-(\eta+c')}{n} -2\gamma(\eta+c').
\end{eqnarray*}
Thus,
\begin{eqnarray*}
&& E[\Inf(S_t)] -E[\Inf(S_{t-1})] \\
&\ge & \frac{\hat{E}[\Inf(S_t)]}{1+\gamma} - \frac{\hat{E}[\Inf(S_{t-1})]}{1-\gamma} \\
&= & \hat{E}[\Inf(S_t)] - \hat{E}[\Inf(S_{t-1})] -\gamma \left( \frac{\hat{E}[\Inf(S_t)]}{1+\gamma} -\frac{\hat{E}[\Inf(S_{t-1})]}{1-\gamma} \right) \\
&\ge & \frac{m-(\eta+c')}{n} -\gamma\left(2+\frac{1}{1+\gamma}\right) (\eta+c').
\end{eqnarray*}
Therefore,
\begin{eqnarray*}
|S_{j-1}\setminus S_i | &\le &
 \frac{E[\Inf(S_{j-1})]-E[\Inf(S_{i})]}{\min_{i<t<j}\{E[\Inf(S_{t})]-E[\Inf(S_{t-1})]\}}\\
&<& (c+c')\cdot\left( \frac{m-(\eta+c')}{n} -\gamma\left(2+\frac{1}{1+\gamma}\right) (\eta+c') \right)^{-1}.
\end{eqnarray*}
Since $\gamma \le \frac{\varepsilon_0(m-(\eta+c'))^2}{8mn(m+\eta n)} \le \frac{\varepsilon_0(m-(\eta+c'))^2}{3n(\eta+c')((c+c')n+\varepsilon_0(m-(\eta+c')))}$,


\begin{eqnarray*}
&& \frac{c+c'}{ \frac{m-(\eta+c')}{n} -\gamma\left(2+\frac{1}{1+\gamma}\right) (\eta+c') } \\
&\le & \frac{(c+c')n}{ m-(\eta+c') -3\gamma (\eta+c')n } \\
&\le & \frac{(c+c')n}{ m-(\eta+c') -\frac{\varepsilon_0(m-(\eta+c'))^2}{(c+c')n +\varepsilon_0(m-(\eta+c'))} } \\
&& \quad \mbox{(by $\gamma\le \frac{\varepsilon_0(m-(\eta+c'))^2}{3n(\eta+c')((c+c')n+\varepsilon_0(m-(\eta+c')))}$)}\\
&= & \frac{(c+c')n((c+c')n+\varepsilon_0(m-(\eta+c')))}{(m-(\eta+c'))(c+c')n }\\
&=& \frac{(c+c')n}{m-(\eta+c')} +\varepsilon_0.
\end{eqnarray*}
It means that
\[ |S_{j}\setminus S_i|< \frac{(c+c')n}{m-(\eta+c')} +\varepsilon_0+1.\]
Since $|S_a|\leq |S_j|+1=|S_i|+|S_j\setminus S_i|+1$, we have
\[ |S_a|\leq \left\lceil \ln \left(\frac{(1+\varepsilon_0)\eta n}{m-\eta}\right) \right\rceil |S^*|+\frac{(c+c')n}{m-(\eta+c')} +3+\varepsilon_0. \]
\qedsymbol
\end{proof}

We now discuss the additive term in Inequality~\eqref{eq:approx}.
To make it nontrivial, we need the additive term to be $o(n)$ as $n$ grows.
This means first that the target set size $m$ should be increasing with $n$, which is reasonable.
Then we should have $c+c'=o(m)$ in order to make the additive term $o(n)$.
In the following theorem, we bound $c$ and $c'$ by the variances of the influence coverage
	of $S^*$ and $S_{a-1}$ respectively, and thus linking the above requirement on $c$ and $c'$
	to the requirement on the variances of influence coverages.

\begin{theorem} \label{thm:additive}
For algorithm $\MinSeed[\varepsilon]$ with any parameter $\varepsilon$, we have
\begin{equation} \label{eq:c}
c\leq \sqrt{\frac{\Var(\Inf(S^*))}{P}}.
\end{equation}
If we use $\MCCP[R]$ for function $\ComputeProbability$ and set
	$R \ge \ln(2n^2)/(2\varepsilon^2) $, then
	algorithm $\MinSeed[\varepsilon]$ finds a seed set $S_a$ such that,
	with probability at least $1-1/n$, $\Pr(\Inf(S_a)\ge \eta)\ge P$ and
\begin{equation} \label{eq:capprox}
c'\leq \sqrt{\frac{\Var(\Inf(S_{a-1}))}{1-P-2\varepsilon}}.
\end{equation}
\end{theorem}

\begin{proof}
We first prove Inequality~\eqref{eq:c}.
If $E[\Inf(S^*)]\geq \eta$, by the definition of $c$ we know that $c=0$, and the first inequality holds trivially. Thus, we only consider the situation where $E[\Inf(S^*)]<\eta$.
\begin{eqnarray*}
&&\Pr(\Inf(S^*)\geq \eta)\\
&=&\Pr(\Inf(S^*)-E[\Inf(S^*)]\geq \eta-E[\Inf(S^*)])\\
&\leq& \Pr(|\Inf(S^*)-E[\Inf(S^*)]|\geq \eta-E[\Inf(S^*)])\\
&\leq& \frac{\Var(\Inf(S^*))}{(\eta-E[\Inf(S^*)])^2}.
\end{eqnarray*}
The last inequality comes from Chebyshev's inequality (Fact~\ref{fact:chebyshev}). Since $\Pr(\Inf(S^*)\geq \eta)\geq P$, by solving the above inequality we have that $\eta-E[\Inf(S^*)]\leq \sqrt{\frac{\Var(\Inf(S^*))}{P}}$, that is, $c\leq \sqrt{\frac{\Var(\Inf(S^*))}{P}}$.


Now suppose that we use $\MCCP[R]$ as an approximation to function
	$\ComputeProbability$.
By Lemma~\ref{lem:compProb}, we know that when we set $R \ge  \ln(2n^2)/(2\varepsilon^2)$,
	for any one seed set $S$, with probability at least $1-1/n^2$,
	algorithm $\MCCP[R]$ approximates the true value within error bound $\varepsilon$.
By union bound, we know that, with probability at least $1-1/n$,
	algorithm $\MinSeed[\varepsilon]$ computes probability
	$\prob$ in line~\ref{line:getprob} for {\em all} seed sets $S_1, S_2, \ldots, $ in $\cal S$
	within error bound $\varepsilon$.
Since for $S_a$, its computed probability is at least $P+\varepsilon$, we know that
	with probability $1-1/n$, $\Pr(\Inf(S_a)\ge \eta )\ge P$.

We now derive Inequality~\eqref{eq:capprox}.
If $E[\Inf(S_{a-1})]\leq \eta$, by the definition of $c'$ we know that $c'=0$, and the
	inequality holds trivially.
Thus, we only need to consider the situation that $E[\Inf(S_{a-1})]> \eta$.
By the algorithm we know that the computed probability of $\Pr(\Inf(S_{a-1})\geq \eta)$
	is less than $P+\epsilon$, so with probability at least $1-\frac{1}{n}$ we have  $\Pr(\Inf(S_{a-1})\geq \eta)< P+2\varepsilon$,
	which means that $\Pr(\Inf(S_{a-1})\leq \eta)\geq 1-P-2\varepsilon$. On the other hand,
\begin{eqnarray*}
&&\Pr(\Inf(S_{a-1})\leq \eta)\\
&=&\Pr(E[\Inf(S_{a-1})]-\Inf(S_{a-1})\geq E[\Inf(S_{a-1})]-\eta)\\
&\leq& \Pr(|E[\Inf(S_{a-1})]-\Inf(S_{a-1})|\geq E[\Inf(S_{a-1})]-\eta)\\
&\leq& \frac{\Var(\Inf(S_{a-1}))}{(E[\Inf(S_{a-1})]-\eta)^2}.
\end{eqnarray*}
The last inequality comes from Chebyshev's inequality (Fact~\ref{fact:chebyshev}).
Thus, we have $E[\Inf(S_{a-1})]-\eta \leq \sqrt{\frac{\Var(\Inf(S_{a-1}))}{1-P-2\varepsilon}}$, that is, $c'\leq \sqrt{\frac{\Var(\Inf(S_{a-1}))}{1-P-2\varepsilon}}$.
%
\qedsymbol \end{proof}

Theorem~\ref{thm:additive} shows that the variances of influence coverages of seed sets,
	or more exactly the standard deviations of influence coverages, determine the scale of
	the additive error term of the algorithm $\MinSeed[\varepsilon]$.
If influence coverages concentrate well with small standard deviations, the algorithm would
	have a good additive error term.
Consider the common case where target set size $m=\Theta(n)$, and $\eta$
	is a constant fraction of $m$, and $P$ is a normal probability requirement not too close
	to $0$ or $1$ (e.g. $0.1$ or $0.5$), if we could have $\Var(\Inf(S^*))=O(m)$ and
	$\Var(\Inf(S_{a-1}))=O(m)$, then $c+c' = O(\sqrt{m})$, and the additive error term
	is $O(n/\sqrt{m}) = O(\sqrt{n}) $.
Together with Theorem~\ref{thm:approxratio}, we would know that
\[
|S_a| \le (\ln n + O(1))|S^*| + O(\sqrt{n}).
\]
In the next section, we analytically show that for one-way bipartite graphs indeed
	$c+c' = O(\sqrt{n})$ (when $m=\Theta(n)$).
We also empirically verify that in real-world graphs the standard deviations of influence coverages
	are indeed small, close to $\sqrt{n}$.
Therefore, our algorithm are likely to perform well in practice.

We remark that our theorems in this section can be applied to a class of models with the following characteristics:
\begin{enumerate}
\item the influence coverage function of a seed set (i.e., $\Inf(\cdot)$) is nonnegative, monotone and submodular, thus greedy algorithm gives an $O(\log n)$-approximation ratio for SM-ECG (Theorem~\ref{thm:greedy}) and provides a tight multiplicative ratio.
\item the influence coverage when choosing the whole set of nodes as seeds is the size of the targeted set (i.e., $\Inf(V) = |U|$), which guarantees that the additive error is reasonable.
\end{enumerate}
The above class includes many diffusion models, such as linear threshold model, general threshold model and continuous time diffusion model.

\section{Results on Bipartite Graphs}
\label{sec:bipartite}
In this section, we solve the SM-PCG problem
	on a one-way bipartite graph $G=(V_1,V_2, E)$, where all edges in $E$ are from $V_1$ to $V_2$. 
For the sake of convenience, we just assume that $U=V_2$ in this section.
It is easy to remove this assumption and make $U$ to be any subset of $V_1\cup V_2$.

One-way bipartite graphs provide two significant advantages over general graphs.
First, it allows a dynamic programming method to compute the exact influence coverage distribution
	given any seed set $S$.
Second, it allows a theoretical analysis on the concentration of influence coverages of seed sets.
We illustrate both aspects below.

We first show how to implement exact computation of function $\ComputeProbability$.
We assign indices for nodes in $V_2$: $v_1, v_2,\ldots, v_m$. Let $A(S,i,j)$ denote the probability that seed set $S$ can activate exactly $j$ nodes in the first $i$ nodes of $V_2$: $v_1, \ldots, v_i$, where $j\leq i$.
Let $p(S, v)$ be the probability that $v$ can be activated by $S$. When $i=1$, it is trivial to get $A(S,1,j)$.
When $i>1$, we can use $A(S,i-1,j-1)$ and $A(S,i-1,j)$ to compute $A(S,i,j)$. If $j=0$, it means $v_1, \ldots, v_{i-1}$ and $v_i$ are all inactive. If $0<j<i$, there are two cases: $j$ nodes are activated in the first $i-1$ nodes while $v_i$ is not activated; $j-1$ nodes are activated in the first $i-1$ nodes and $v_i$ is activated. If $j=i$, both $v_1,\ldots,v_{i-1}$ and $v_i$ are activated. Thus, we have the following recursion,
\begin{displaymath}
A(S,1,j)=\left\{
              \begin{array}{ll}
                p(S, v_1), & \hbox{$j=1$} \\
                1-p(S, v_1), & \hbox{$j=0$}
              \end{array}
            \right.
\end{displaymath}
and
\begin{displaymath}
A(S, i, j) =  \left\{
              \begin{array}{ll}
                A(S,i-1,j)\cdot(1-p(S, v_i)), & \hbox{$j=0$} \\
                A(S,i-1,j)\cdot(1-p(S, v_i)) & \\
                +A(S,i-1,j-1)\cdot p(S, v_i), & \hbox{$0<j<i$} \\
                A(S,i-1,j-1)\cdot p(S, v_i), & \hbox{$j=i$}
              \end{array}
            \right.
\end{displaymath}
For IC model, $p(S, v_i) = 1-\prod_{u\in S}(1-p_{u,v_i})$; and for LT model, $p(S, v_i) = \sum_{u\in S} p_{u, v_i}$.
Using the above dynamic programming formulation, we can implement
	function $\ComputeProbability$ as function $\BiCP$ given
	in Algorithm \ref{alg:BiCP}.

\begin{algorithm}[thb]
\renewcommand{\algorithmicrequire}{\textbf{Input:}}
\renewcommand\algorithmicensure {\textbf{Output:}}
    \caption{Function $\BiCP$ for bipartite graphs} \label{alg:BiCP}
    \begin{algorithmic}[1]
    \REQUIRE ~~ $G=(V_1,V_2,E),\{p_{u,v}\}_{(u,v)\in E}, S, \eta$\\
    \ENSURE ~~  $P=\Pr(\Inf(S)\geq \eta)$\\
\FOR {$i$ from 1 to $n$, and $j$ from $1$ to $i$}
    \STATE \ compute $A(S,i,j)$ via dynamic programming\\
\ENDFOR
\RETURN $\sum_{j=\eta}^{m} A(S,m,j)$
\end{algorithmic}
\label{probabilityonbipartite}
\end{algorithm}

One-way bipartite graphs have an important property that the activation events of nodes
	in $V_2$ are mutually independent.
This allows us to bound  $c$ and $c'$ defined in Section 5 using Hoeffding's Inequality, as
	shown in the following theorem.

\begin{theorem}
For algorithm $\MinSeed[0]$ on one-way bipartite graph $G=(V_1,V_2,E)$ , we have
\[
c\leq \sqrt{\frac{m}{2}\ln\frac{1}{P}}, c'\leq \sqrt{\frac{m}{2}\ln \frac{2}{1-P}}.
\]
\end{theorem}

\begin{proof}
Suppose $X_1,\ldots,X_m$ are random variables corresponding to
	nodes $v_1,v_2,\ldots, v_m$ in $V_2$, such that for each $1\leq i \leq m$, $X_i=1$ if $v_i$ is activated and $X_i=0$ otherwise.
Since $G$ is a bipartite graph and all edges in $E$ are from $V_1$ to $V_2$, after nodes in $V_2$
	are activated, they will not continue to influence other nodes.
It means that for each $v_i\in V_2$, whether $v_i$ is activated is independent with the activations of other nodes in $V_2$. In other words, $X_i$ are independent random variables.

Similar to the proof of Theorem~\ref{thm:additive}, we only need to consider the situation where $E[\Inf(S^*)]<\eta$ and $E[\Inf(S_{a-1})]>\eta$.
\begin{eqnarray*}
&&\Pr(\Inf(S^*)\geq \eta )\\
&=&\Pr(\Inf(S^*)-E(\Inf(S^*))\geq \eta -E(\Inf(S^*)))\\
&\leq& \exp \left(-\frac{2(\eta -E(\Inf(S^*)))^2}{m}\right).
\end{eqnarray*}
The last inequality comes from Hoeffding's inequality (Fact~\ref{fact:hoeffding}).
Since $\Pr(\Inf(S^*)\geq\eta )\geq P$, by solving the above inequality, we get $\eta-E[\Inf(S^*)]\leq \sqrt{\frac{m}{2}\ln \frac{1}{P}}$, that is, $c\leq \sqrt{\frac{m}{2}\ln \frac{1}{P}}$.

%
On the other hand, we know that $\Pr(\Inf(S_{a-1})\geq \eta)<P$, as well as
\begin{eqnarray*}
&&\Pr(\Inf(S_{a-1})\leq \eta)\\
&=&\Pr(E(\Inf(S_{a-1}))-\Inf(S_{a-1})\geq E(\Inf(S_{a-1})-\eta))\\
&\leq& \Pr(|E(\Inf(S_{a-1}))-\Inf(S_{a-1})|\geq E(\Inf(S_{a-1})-\eta)\\
&\leq& 2exp(-\frac{2(E(\Inf(S_{a-1})-\eta )^2}{m}).
\end{eqnarray*}
The last inequality comes from Hoeffding's inequality (Fact~\ref{fact:hoeffding}). Since $\Pr(\Inf(S_{a-1})\leq \eta)\geq 1-P$, we have that $\eta-E[\Inf(S_{a-1})]\leq \sqrt{\frac{m}{2}\ln \frac{2}{1-P}}$, that is, $c'\leq \sqrt{\frac{m}{2}\ln \frac{2}{1-P}}$.
\qedsymbol \end{proof}

Together with Theorem \ref{thm:approxratio} we get the following corollary.

\begin{corollary}
For one-way bipartite graphs, algorithm $\MinSeed[0]$ using function $\BiCP$ returns seed
	set $S_a$ such that $\Pr(\Inf(S_a)\ge \eta) \ge P$, and when we consider the probability
	threshold $P$ as a constant independent of $n$ and $m$, we have
$$|S_a|\le (\ln n + O(1)) |S^*|+O(\frac{n}{\sqrt{m}}).$$
\end{corollary}


We note that one-way bipartite graphs are a restricted class of graphs, where the influence cascading is a 1-hop cascading process and cannot be generated to a cascade with greater depth.
However, we believe their
	analytical results can shed lights on more realistic networks when most of
	node activations in the network are independent.

\section{Experiments}
\label{sec:exp}
We conduct experiments on real social networks for the  following purposes: (1) test the concentration of influence coverage distributions of seed sets; (2) validate the performance
	of our algorithm against baseline algorithms.

\subsection{Experiment setup}
\textbf{Datasets.} We conduct experiments on three real social networks. The first one is wiki-Vote, published by Leskovec~\cite{website:wiki-Vote}. It is a network relationship graph from Wikipedia community, with totally 7,115 nodes and 103,689 edges. In wiki-Vote graph, each node represents a user in Wikipedia community, and an edge $(u,v)$ represents user $u$ votes for user $v$, which means that $v$ has an influence on $u$. Thus, in our experiment, we reverse all edges to express the influence between pairs of nodes.
We use weighted cascade (WC) model \cite{kempe2003maximizing}
to assign the influence probabilities on edges. For each edge $(u,v)$, we assign its probability to be $1/d_{in}(v)$, where $d_{in}(v)$ is the in-degree of node $v$.

The second network is NetHEPT, which is a standard dataset used in \cite{ChenWY09,chen2010scalable,ChenYZ10,simpath,goyal2012minimizing}.
NetHEPT is an academic collaboration network from arXiv (http://www.arXiv.org), with totally 15,233 nodes and 58,891 edges. In NetHEPT graph, each node represents an author, and each edge represents coauthor relationship between two authors. NetHEPT is an undirected graph, and in our experiment we add two directed edges between two nodes if there exists at least one edge between these two nodes in NetHEPT.
Similar to wiki-Vote, we use WC model to assign edge influence probabilities.
We assign the probability on directed edge $(u,v)$ to be $d(u,v)/d(v)$, where $d(u,v)$ is the number of papers collaborated by $u$ and $v$, and $d(v)$ is the number of papers published by $v$.

The last one is Flixster, an American movie rating social site. Each node is a user, and edges describe the friendship between users. In this network, we use a Topic-aware Independent Cascade Model from \cite{barbieri2012topic} to learn the real influence probabilities on edges for different topics. We simply use two different topics, say topic 1 and topic 2, and get the edge probabilities that one user influences his/her friend on the specific topic. In both topics, we remove edges with probability 0 and isolated nodes. For topic 1, there are 28,317 nodes and 206,012 edges. The mean of edge probabilities is 0.103, and the standard deviation is 0.160. For topic 2, there are 25,474 nodes and 135,618 edges. The mean of edge probabilities is 0.133, and the standard deviation is 0.205.


\textbf{Experiment methods.} In the experiment, for the sake of convenience, we set $U=V$.

Our first task is to test the concentration of influence coverage distributions of seed sets.
To do so, we test the variances (or their square roots, i.e. standard deviations).
According to Theorem~\ref{thm:additive}, small standard deviations imply small
	$c$ and $c'$ and thus small additive errors
	of the $\MinSeed[\varepsilon]$ algorithm output.
By Inequality~(\ref{eq:capprox}), to verify that $c'$ is small, we just need to
	test the standard deviations of all seed sets generated by the algorithm.
For quantity $c$, we need to test the standard deviation of the influence coverage of the
	optimal seed set, according to Inequality~(\ref{eq:c}).
However, finding the optimal seed set is NP-hard, therefore we cannot fully verify
	the bound on $c$.
To compensate, we test randomly selected seed sets as follows.
For each fixed seed set size $k$, we independently select 10 seed sets of size $k$ at random,
	and compute the maximum standard deviations of the influence coverage distributions
	of these selected seed sets.
Although randomly selected seed sets may be far from the optimal seed set, what we hope is that
	by testing standard deviations on both randomly selected sets and greedily selected sets
	by algorithm $\MinSeed[\varepsilon]$, we have a general understanding of standard deviations
	of influence coverages of seed sets, which may provide us with hints for other seed sets, such as the optimal seed
	set.
To estimate the standard deviations of influence coverage of a seed set $S$, we use 10,000 times Monte Carlo simulation and compute the variance, and
	take its square root to obtain the standard deviation.

Our second task is to test the performance of seed selection algorithm
	$\MinSeed[\varepsilon]$.
We compare the performance with three baseline algorithms: (a) {\sf Random},
	which generates the seed set sequence in random order;
	(b) {\sf High-degree}, which generates the seed set sequence according
	to the decreasing order of the out-degree of nodes; and
	(c) {\sf PageRank}, which is a popular method for website ranking~\cite{brin1998anatomy}. We use 	$p_{v,u}/ \sum_{(w,u)\in E} p_{w,u}$ as the transition probability for edge $(u,v)$. Higher 					$p_{v,u}$ means that $v$ is more influential to $u$, indicating that $u$ ranks $v$ higher. We use 		0.15 as the restart probability and use the power method to compute PageRank values. When two 		consecutive iterations are different for at most $10^{-4}$ in $L_1$ norm, we stop.
As for our $\MinSeed[\varepsilon]$ algorithm, to speed up the algorithm,
	we use the state-of-the-art PMIA algorithm of \cite{chen2010scalable} to greedily generate the seed set sequence.
For all the above algorithms, we use the same $\MCCP[R]$
	algorithm to compare whether
	a seed set $S$ in the sequence satisfies the condition
	$\Pr(\Inf(S)\ge \eta) \ge P+\varepsilon$.
Since the seed set sequence generations in all the above algorithms are fast
	comparing to the Monte Carlo simulation based $\MCCP[R]$ algorithm,
	our implementation actually generates the sequence first and then uses binary
	search to find the seed set in the sequence satisfying
	$\Pr(\Inf(S)\ge \eta) \ge P+\varepsilon$.

We set parameters $R=10,000$ and $\varepsilon = 0.01$.
One may see that these settings do not satisfy the condition
	$R \ge  \ln(2n^2)/(2\varepsilon^2)$ in
	Theorem~\ref{thm:additive} for our datasets: in our datasets, $n$
	is around $10^4$, and thus $\ln(2n^2)/(2\varepsilon^2)$ is
	around $9.6 \times 10^4$.
However, we can justify our choice as follows.
First, the condition $R \ge  \ln(2n^2)/(2\varepsilon^2)$ is a conservative
	theoretical condition for obtaining high probability of $1-1/n$ for our
	approximation guarantee.
In practice, a smaller $R$ of $10,000$ is good enough for illustrating our
	results.
Second, all algorithms use the same $\MCCP[R]$ algorithm, so the comparison
	is fair among them, and is focused on the difference in
	their generations of seed set sequences,
	not on the accuracy of the estimate of function $\ComputeProbability$.
Third, the seed selections actually depends only on the combined
	parameter $P' = P+\varepsilon$, and not on $P$ and $\varepsilon$
	separately.
Thus setting $\varepsilon = 0.01$ is only for
	intuitive understanding and setting it to some other value would not
	change the results as long as $P'$ remains the same.

\subsection{Experiment results}


\begin{figure}[t]
\centering
\begin{minipage}[t]{0.45\linewidth}
\includegraphics[width=\linewidth]{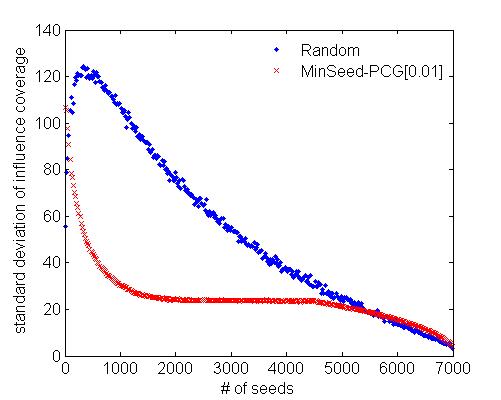} \\
\vspace{-.7cm}
{\center (a) wiki-Vote graph \\}
\end{minipage}
\begin{minipage}[t]{0.45\linewidth}
\includegraphics[width=\linewidth]{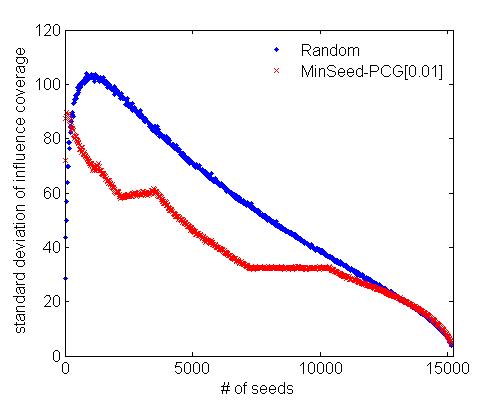}
\vspace{-.7cm}
{\center (b) NetHEPT graph \\}
\end{minipage}
\begin{minipage}[t]{0.45\linewidth}
\includegraphics[width=\linewidth]{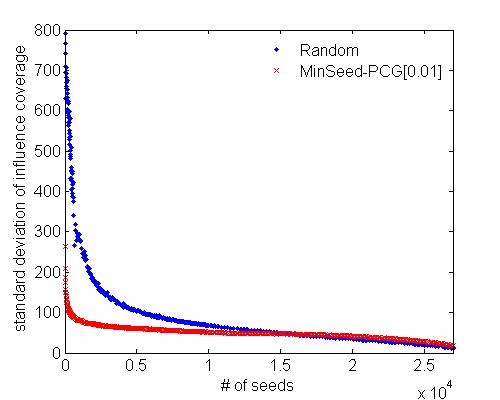} \\
\vspace{-.7cm}
{\center (c) Flixster graph with topic 1 \\}
\end{minipage}
\begin{minipage}[t]{0.45\linewidth}
\includegraphics[width=\linewidth]{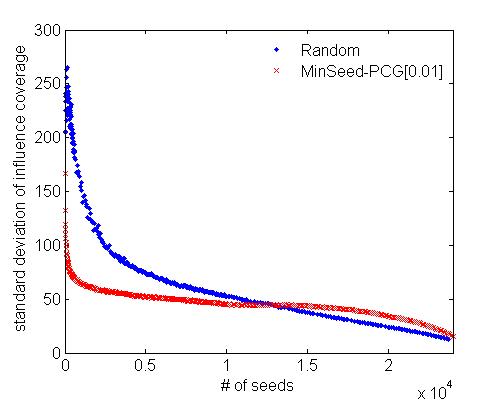} \\
\vspace{-.7cm}
{\center (d) Flixster graph with topic 2 \\}
\end{minipage}
\caption{Standard deviations of influence coverages of
	seed sets.}
\label{fig:standardd}
\end{figure}

\textbf{Concentration of influence coverages.}
Figure~\ref{fig:standardd} shows the standard deviations of
	influence coverages of randomly selected seed sets
	and greedily selected seed sets (by algorithm $\MinSeed[\varepsilon]$) on wiki-Vote, NetHEPT and Flixster.
We can see that in all graphs, standard deviations for greedily selected
	seed sets quickly drop, while for randomly selected seed sets sometimes it
	has a small increase when the seed set size is small, and then quickly drop
	too.
The maximum value is about $130$ for wiki-Vote ($|V|=7,115$), $105$ for
	NetHEPT ($|V|=15,233$), $760$ for Flixster with topic 1 ($|V|=28,317$),
	and $270$ for Flixster with topic 2 ($|V|=25,474$).
Thus by observation the standard deviation is at the order of $\sqrt{|V|}$.
As discussed after Theorem~\ref{thm:additive}, this means that
	the additive error of our algorithm would be $O(\sqrt{|V|})$, a small
	and satisfactory value.
The standard deviations for wiki-Vote are larger
	than those for NetHEPT at small seed set size even though the number of nodes
	of wiki-Vote is smaller.
We believe this is because wiki-Vote has more edges (103,689) than NetHEPT (58,891), and thus when the seed
	set size is small more edges could cause larger variances in influence
	coverage. This can also explain why in Flixster topic 1 (with 206,012 edges) has larger standard deviations than topic 2 (with 135,618 edges).

\begin{figure}[t]
\centering
\begin{minipage}[t]{0.45\linewidth}
\includegraphics[width=\linewidth]{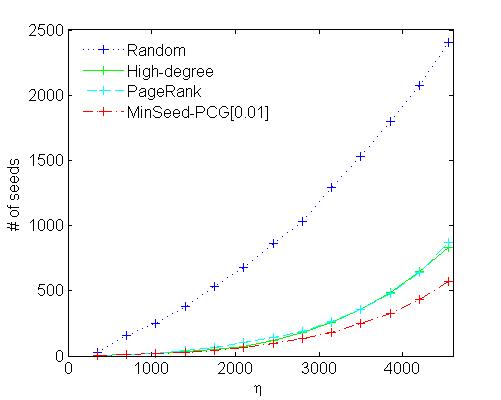} \\
\vspace{-.7cm}
{\center (a) wiki-Vote graph \\}
\end{minipage}
\begin{minipage}[t]{0.45\linewidth}
\includegraphics[width=\linewidth]{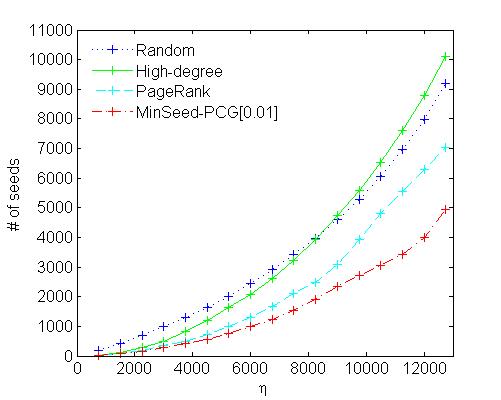} \\
\vspace{-.7cm}
{\center (b) NetHEPT graph \\}
\end{minipage}
\begin{minipage}[t]{0.45\linewidth}
\includegraphics[width=\linewidth]{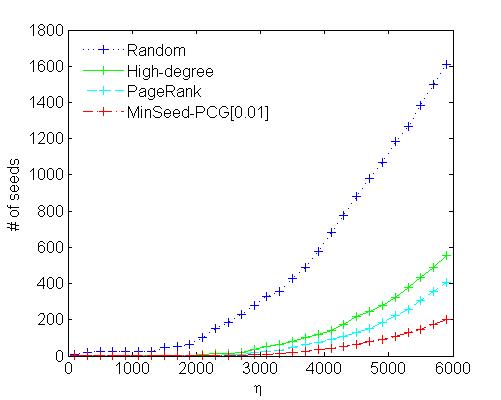} \\
\vspace{-.7cm}
{\center (c) Flixster graph with topic 1 \\}
\end{minipage}
\begin{minipage}[t]{0.45\linewidth}
\includegraphics[width=\linewidth]{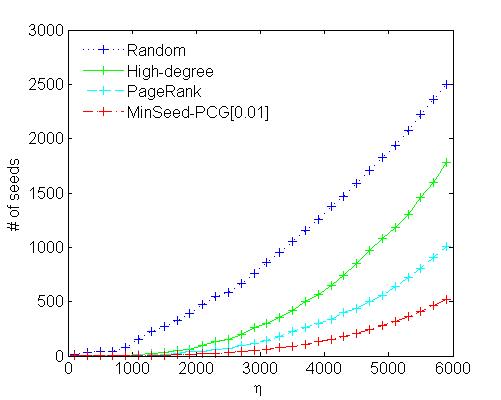} \\
\vspace{-.7cm}
{\center (d) Flixster graph with topic 2 \\}
\end{minipage}
\caption{Size of selected seed sets vs. coverage threshold $\eta$ under
	a fixed probability threshold $P=0.1$.}
\label{fig:seedsvseta}
\end{figure}

\begin{figure}[t]
\centering
\begin{minipage}[t]{0.45\linewidth}
\includegraphics[width=\linewidth]{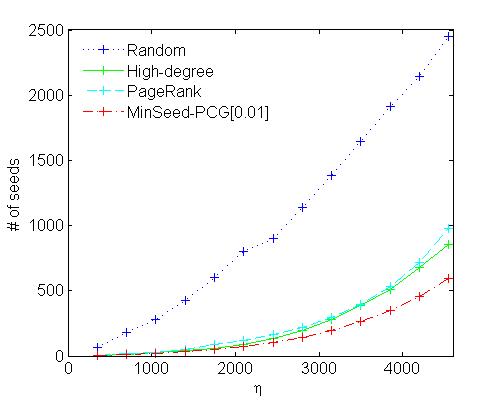} \\
\vspace{-.7cm}
{\center (a) wiki-Vote graph \\}
\end{minipage}
\begin{minipage}[t]{0.45\linewidth}
\includegraphics[width=\linewidth]{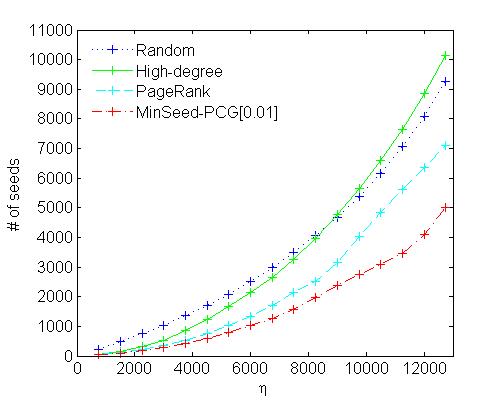} \\
\vspace{-.7cm}
{\center (b) NetHEPT graph \\}
\end{minipage}
\begin{minipage}[t]{0.45\linewidth}
\includegraphics[width=\linewidth]{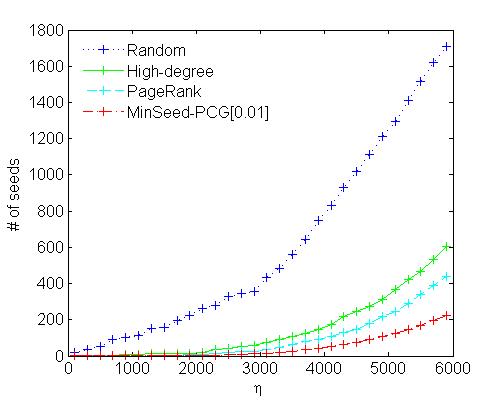} \\
\vspace{-.7cm}
{\center (c) Flixster graph with topic 1 \\}
\end{minipage}
\begin{minipage}[t]{0.45\linewidth}
\includegraphics[width=\linewidth]{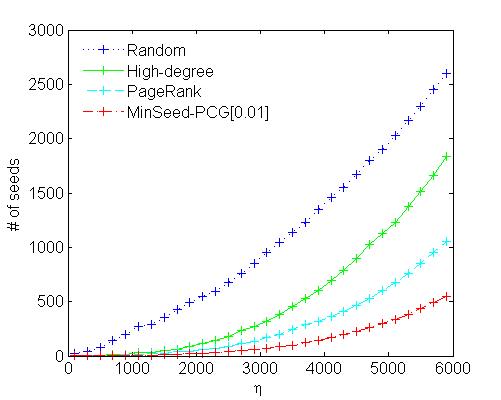} \\
\vspace{-.7cm}
{\center (d) Flixster graph with topic 2 \\}
\end{minipage}
\caption{Size of selected seed sets vs. coverage threshold $\eta$ under
	a fixed probability threshold $P=0.5$.}
\label{fig:seedsvseta2}
\end{figure}


\begin{figure}[t]
\centering
\begin{minipage}[t]{0.45\linewidth}
\includegraphics[width=\linewidth]{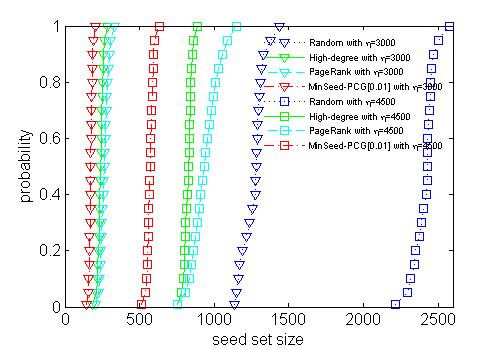}\\
\vspace{-.7cm}
{\center (a) wiki-Vote graph, $\eta=3000, 4500$ \\}
\end{minipage}
\begin{minipage}[t]{0.45\linewidth}
\includegraphics[width=\linewidth]{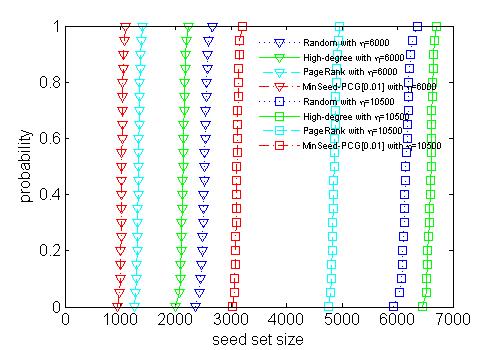}\\
\vspace{-.7cm}
{\center (b) NetHEPT graph, $\eta=6000, 10500$ \\}
\end{minipage}
\begin{minipage}[t]{0.45\linewidth}
\includegraphics[width=\linewidth]{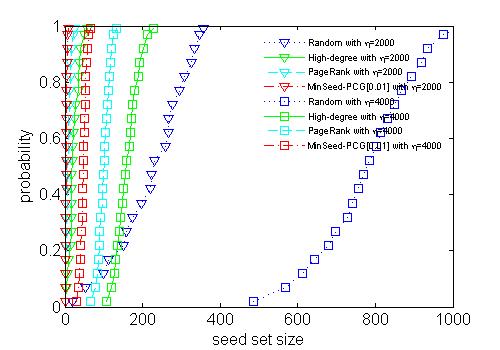}\\
\vspace{-.7cm}
{\center (c) Flixster graph with topic 1, $\eta=2000, 4000$ \\}
\end{minipage}
\begin{minipage}[t]{0.45\linewidth}
\includegraphics[width=\linewidth]{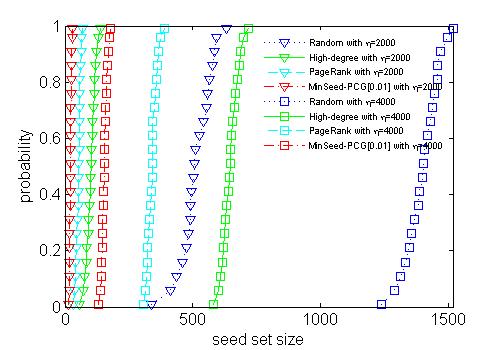}\\
\vspace{-.7cm}
{\center (c) Flixster graph with topic 2, $\eta=2000, 4000$ \\}
\end{minipage}
\caption{Size of selected seed sets vs. probability threshold $P$ under a fixed coverage threshold $\eta$.}
\label{fig:probvsseeds}
\end{figure}


\textbf{Performance of $\MinSeed[\varepsilon]$ compared with baselines.}
We conduct two sets of tests for this purpose.
First, we fix the probability threshold $P$ to $0.1$ and $0.5$, and vary
	the coverage threshold $\eta$ to compare the size of seed sets selected
	by various algorithms.
Figure~\ref{fig:seedsvseta} and Figure~\ref{fig:seedsvseta2} show the test results on three datasets.
All test results consistently show that our algorithm performances the best,
	and sometimes with a significant improvement over the {\sf Random}
	and {\sf High-degree} heuristics.
In particular, for wiki-Vote and $P=0.1$ (Figure~\ref{fig:seedsvseta}(a)), on average
	our algorithm $\MinSeed[\varepsilon]$
	selects seed sets with size $88.2\%$ less than those selected by
	{\sf Random}, $20.2\%$ less than {\sf High-degree}, and $30.9\%$ less than {\sf PageRank}.
For NetHEPT and $P=0.1$ (Figure~\ref{fig:seedsvseta}(b)),
	on average our algorithm selects seed sets
	with size $56.7\%$ less than {\sf Random}, $46.0\%$
	less than {\sf High-degree}, and $24.4\%$ less than {\sf PageRank}.
The {\sf High-degree} heuristic performs close to $\MinSeed[\varepsilon]$
	in wiki-Vote, but performs badly in NetHEPT, even worse than
	{\sf Random} when $\eta$ is large.
This shows that {\sf High-degree} is not a good and stable heuristic for
	this task.
For Flixster with topic 1 and $P=0.1$ (Figure~\ref{fig:probvsseeds}(c)), on average $\MinSeed[\varepsilon]$ selects seed sets with size $94.4\%$ less than {\sf Random}, $54.0\%$ less than {\sf High-degree}, and $29.2\%$ less than {\sf PageRank}. For Fixster with topic 2 and $P=0.1$ (Figure~\ref{fig:probvsseeds}(d)), on average $\MinSeed[\varepsilon]$ selects seed sets with size $91.1\%$ less than {\sf Random}, $73.0\%$ less than {\sf High-degree}, and $24.4\%$ less than {\sf PageRank}.
Figures~\ref{fig:seedsvseta2} show the results for
	$P=0.5$.
The curves are almost the same as the corresponding ones for $P=0.1$.
This can be explained by the sharp phase transition to be
	observed in the next set of tests, which is due to concentration of influence
	coverage, such that typically only a few tens of more seeds would
	satisfy probability threshold $P$ from $0.1$ to $0.5$.
	
Our second set of tests is to fix a coverage threshold $\eta$, and
	observe the change
	of coverage	probability $\Pr(\Inf(S)\ge \eta)$ as the seed set $S$
	grows as computed by various algorithms.
Figure~\ref{fig:probvsseeds} shows the test results for the three datasets.
Wiki-Vote, NetHEPT and Flixster with topic 2 (Figure~\ref{fig:probvsseeds}(a), (b), (d)) have sharp phase transition: there is a short range of seed set size where the probability increases very fast from 0.01 to very close to 1 (only several nodes are needed to reach a 0.1 increment in probability). While Filxster with topic 1 (Figure~\ref{fig:probvsseeds}(c)) has a relatively smooth phase transition. This phase transition phenomenon is clearly due to the concentration of influence coverages of seed sets, as already verified in Figure~\ref{fig:standardd}.

	

In all our tests, $\MinSeed[\varepsilon]$ performances the best:
	its phase transition comes first before the other algorithms, which means
	it uses less number of seeds to achieve the same probability threshold
	$P$.
{\sf Random} performs much worse than $\MinSeed[\varepsilon]$, while
	{\sf PageRank} and {\sf High-degree} perform close to $\MinSeed[\varepsilon]$ when
	$\eta$ is small, but noticeably worse than $\MinSeed[\varepsilon]$
	when $\eta$ gets larger.
For wiki-Vote graph, on average $\MinSeed[\varepsilon]$ selects a seed set with size $34.1\%$ less than {\sf PageRank}, $27.7\%$ less than {\sf High-degree}, and $86.4\%$ less than {\sf Random} when $\eta=3,000$. When choose $\eta=4,500$, $\MinSeed[\varepsilon]$ selects a seed set with size on average $38.8\%$ less than {\sf PageRank}, $30.8\%$ less than {\sf High-degree}, and $76.3\%$ less than {\sf Random}. For NetHEPT graph, when $\eta=6,000$, on average $\MinSeed[\varepsilon]$ selects a seed set with size $22.8\%$ less than {\sf PageRank}, $51.8\%$ less than {\sf High-degree}, and $59.2\%$ less than {\sf Random}. When $\eta=10,500$, on average $\MinSeed[\varepsilon]$ selects a seed set with size $36.1\%$ less than {\sf PageRank}, $52.9\%$ less than {\sf High-degree}, and $49.6\%$ less than {\sf Random}.
For Flixster graph with topic 1, when $\eta=2,000$, on average the output number of seeds by $\MinSeed[\varepsilon]$ is $44.1\%$ less than {\sf PageRank}, $78.9\%$ less than {\sf High-degree}, and $98.3\%$ less than {\sf Random}. When $\eta=4,000$, the corresponding results are $53.2\%$,  $70.7\%$ and $93.9\%$. For topic 2, when $\eta=2,000$, on average the output number of seeds by $\MinSeed[\varepsilon]$ is $59.0\%$ less than {\sf PageRank}, $78.6\%$ less than {\sf High-degree}, and $95.8\%$ less than {\sf Random}. When $\eta=4,000$, the corresponding results are $54.9\%$, $76.2\%$ and $89.0\%$.

For all these graphs, we do not test the case
	when $\eta$ is very close to the number of nodes.
Since in this case a large seed set close to the full node set is needed,
	and greedy-based seed selection loses its advantage comparing to
	simple random or high-degree heuristics when a large number of seeds
	are needed.
Moreover, we believe that requiring $\eta$ to be close to the full network size
	is not a realistic scenario in practice.

As a summary, our experimental results validate that influence coverages of
	seed sets are concentrated well in real-world networks, and thus
	support the claim that our algorithm provides good approximation guarantee.
Moreover, our algorithm performs much better than simple baseline algorithms,
	achieving significant savings on seed set size.

\section{Future Work}
\label{sec:conclusion}


This study may inspire a number of future directions.
One is to study the concentration property of other classes of graphs, especially
	graphs close to real-world networks such as power-law graphs, to see if we
	can analytically prove that a large class of graphs have good concentration
	property on influence coverage distributions.
Another direction is to speed up the estimation of $\Pr(\Inf(S)\ge \eta)$, which
	is done by Monte Carlo simulation in this work and is slow.
One may also study influence maximization problem where
	reaching the tipping point is the first step, which is followed by further diffusion
	steps.
Our algorithm and results may be an integral component of such influence maximization
	tasks.



%
\bibliographystyle{abbrv}
\bibliography{draft}  
%
%
\clearpage
\appendix
\section*{Appendix: Bipartite Graphs for Full Coverage}
In this appendix, we will discuss about SM-PCG problem with $\eta =|U|$ on a one-way bipartite graph $G=(V_1, V_2, E)$ where all edges are from $V_1$ to $V_2$. For the sake of convenience, we assume that $U=V_2$, which is easy to be removed. Let $V = V_1\cup V_2$, and $|V| = n$, $|U| = m$.
We propose an $O(\log m)$-approximation algorithm for both edge probabilities and probabilistic threshold $P$ being constant, which asymptotically matches the inapproximation result from Theorem~\ref{thm:hard}.

This algorithm is described in Algorithm~\ref{two-stage}, which contains two stages.
Firstly, we greedily select a seed set, say $S_1$, such that all nodes in $U$ can be reached from nodes in $S_1$. In Algorithm~\ref{two-stage}, we define $R(S)$ to be a subset of $U$ such that for each $u\in R(S)$, there exists an edge $(s,u)\in E$ for some $s\in S$. Intuitively, it is to find a set cover greedily where the universe is $U$ and the collection of subsets is $V$. Secondly, based on the selected $S_1$, we define a set function $f_{S_1}(X): 2^{V\setminus S_1}\rightarrow (0,1]$, which computes the probability that $X\cup S_1$ activates all nodes in $U$.
Unfortunately, it is easy to verify that $f_{S_1}$ is nonsubmodular. We define another function $g_{S_1}(X)=\log f_{S_1}(X)-\log f_{S_1}(\emptyset)$. Obviously, $g_{S_1}(X)$ is non-negative and monotone. Actually, we will show that $g_{S_1}(X)$ is also submodular later. Based on these nice properties of $g_{S_1}(X)$,
we use greedy algorithm to find another seed set $S_2$ such that $g_{S_1}(S_2)\geq \log P-\log f_{S_1}(\emptyset)$, that is, $S_1\cup S_2$ can activate all nodes in $U$ with a probability at least $P$.

\begin{algorithm}[thb]
\renewcommand{\algorithmicrequire}{\textbf{Input:}}
\renewcommand\algorithmicensure {\textbf{Output:}}
    \caption{Two-stage Algorithm}
    \begin{algorithmic}[1]
    \REQUIRE ~~ $G=(V_1,V_2,E),P$\\
    \ENSURE ~~ $S=\arg\min_{S':\Pr(\Inf(S')=m)\geq P}\{|S'|\}$
\STATE set $S_1=\emptyset$, $S_2=\emptyset$
\STATE /*first-stage*/
\WHILE {$\exists u\in U$, there is no $s\in S$ such that $(s,u)\in E$}
    \STATE \ select $v=\arg\max_{w\in V}\{R(S_1\cup \{w\})-R(S_1)\}$
    \STATE \ $S_1=S_1\cup \{v\}$
\ENDWHILE
\STATE /*second-stage*/
\WHILE {$g_{S_1}(S_2)<\log P-\log f(S_1)$}
    \STATE \ select $v=\arg\max_{w\in V\setminus S_1}\{g_{S_1}(S_2\cup \{w\})-g_{S_1}(S_2)\}$
    \STATE \ $S_2=S_2\cup \{v\}$
\ENDWHILE
\RETURN $S_1\cup S_2$
\end{algorithmic}
\label{two-stage}
\end{algorithm}


Let $S_1^*$ be an optimal set in the first stage, and $S_2^*$ be an optimal set in the second stage based on $S_1$. Let $S^*$ denote an optimal seed set with the minimum size such that the probability of activating $U$ is at least $P$. For set cover problem, greedy algorithm provides an $\ln m-\ln \ln m + \Theta(1)$ approximation~\cite{slavik1996tight}. Thus, it is easy to see that $|S_1|\leq \ln m |S_1^*|$. We will show that $g_{S_1}(X)$ is submodular, which indicates greedy algorithm in the second stage also provides a good approximation guarantee.

\begin{lemma}
$g_{S_1}(X)$ is submodular.
\end{lemma}
\begin{proof}
Suppose for any two sets $T_1\subseteq T_2 \subseteq V\setminus S_1$, and any $u\in V\setminus (S_1\cup T_2)$.
Let $N(u)$ be the set of all out-neighbors of $u$, and let $P(T,v)$ be the probability that $T$ influence $v$. Then, we have the following result,
\begin{eqnarray*}
&&g_{S_1}(T_1\cup \{u\})-g_{S_1}(T_1)\\
&=&\sum_{v\in N(u)} \left( \log P(T_1\cup S_1\cup \{u\}, v) - P(T_1\cup S_1, v) \right) \\
&=&\sum_{v\in N(u)} \left( \log (P(T_1\cup S_1, v)+P(\{u\},v) -P(T_1\cup S_1, v)\times P(\{u\},v))-\log P(T_1\cup S_1, v) \right)\\
&=&\sum_{v\in N(u)} \log \left(1+\frac{P(\{u\},v)}{P(T_1\cup S_1, v)}-P(\{u\},v)\right)
\end{eqnarray*}
Similarly, we have
$$g_{S_1}(T_2\cup \{u\})-g_{S_1}(T_2)=\sum_{v\in N(u)} \log \left(1+\frac{P(\{u\},v)}{P(T_2\cup S_1, v)}-P(\{u\},v)\right)$$
Since $T_1\subseteq T_2$, $P(T_1\cup S_1, v)\leq P(T_2\cup S_1, v)$. It means that $g_{S_1}(T_1\cup \{u\})-g(T_1)\geq g(T_2\cup \{u\})-g(T_2)$, thus, $g_{S_1}(X)$ is submodular.
\qedsymbol \end{proof}

\begin{theorem}
Our two-stage greedy algorithm provides a seed set $S$ with the size
$\left(\ln m+\left\lceil \ln \left(m\left(\frac{m\log p_{min}}{\log P}-1\right)\right) \right\rceil\right)\cdot |S^*|+1$, where $p_{min}$ is the smallest edge probability on $G$.
\end{theorem}

\begin{proof}
Since $g_{S_1}(X)$ is monotone and submodular, by theorem \ref{bicriteria}, we can find a seed set $S_2'$ such that $g_{S_1}(S_2')\geq \log P-\log f_{S_1}(\emptyset)-\varepsilon$ and $|S_2'|\leq |S_2^*|\cdot\left\lceil \ln \left(\frac{\log P-\log f_{S_1}(\emptyset)}{\varepsilon}\right) \right\rceil$, where we set $\varepsilon=-\log \sqrt[m]{P}$. If $g_{S_1}(S_2')\ge \log P-\log f_{S_1}(\emptyset)$, we set $S_2=S_2'$; otherwise, we find the node $v=\arg\min_{v_i\in V}\{P(S_1\cup S_2',v_i)\}$, and add $v$ into seed set, that is, $S_2=S_2'\cup \{v\}$. Now, we have
\begin{eqnarray}
g_{S_1}(S_2)&=&\log f_{S_1}(S_2)-\log f_{S_1}(\emptyset)\\
&=&\log f_{S_1}(S_2')-\log P(S_2',v)-\log f_{S_1}(\emptyset) \label{eq1}\\
&\geq& \log P + \log \sqrt[m]{P}-\log P(S_2',v)-\log f_{S_1}(\emptyset) \label{eq2}\\
&=&\log P + \log \frac{\sqrt[m]{P}}{P(S_2',v)}-\log f_{S_1}(\emptyset)\\
&\geq& \log P-\log f_{S_1}(\emptyset). \label{eq3}
\end{eqnarray}
Equality~\eqref{eq1} comes from the independence of the activation of nodes in $U$.
Inequality~\eqref{eq2} comes from the fact $g_{S_1}(S_2')\geq \log P-\log f_{S_1}(\emptyset)+\log \sqrt[m]{P}$. And~\eqref{eq3} holds since $P(S_2',v)$ is the minimum activation probability for nodes in $U$, which is smaller than the average activation probability $\sqrt[m]{P}$.
On the other hand, we have
\begin{eqnarray}
|S_2|&\leq& |S_2^*|\cdot\left\lceil \ln \left(\frac{\log P-\log f_{S_1}(\emptyset)}{-\log \sqrt[m]{P}}\right) \right\rceil+1\\
&\leq& |S_2^*|\cdot\left\lceil \ln \left(\frac{m\log P- m^2 \log p_{min}}{-\log P}\right) \right\rceil+1\label{eq4}\\
&=&|S_2^*|\cdot\left\lceil \ln \left( m\left(\frac{m\log p_{min}}{\log P}-1\right)\right) \right\rceil+1.
\end{eqnarray}
Inequality~\eqref{eq4} comes from that $f_{S_1}(\emptyset)=\prod _{u\in U}P(S_1,u)\geq p_{min}^m$.

Since we require that all nodes in $U$ to be activated, thus  $S^*$ is a feasible solution in the first stage. It means that $|S_1|\le \ln m|S_1^*|\leq \ln m|S^*|$. On the other hand, since $g_{S_1}(S^*)=\log f_{S_1}(S^*)-\log f_{S_1}(\emptyset)\geq \log P-\log f_{S_1}(\emptyset)$, thus $|S_2^*|\leq |S^*|$. So, we have shown that
$$|S|=|S_1|+|S_2|\leq \left(\ln m  +\left\lceil \ln \left(m\left(\frac{m\log p_{min}}{\log P}-1\right)\right) \right\rceil \right)\cdot |S^*|+1.$$
\qedsymbol \end{proof}
When both $p_{min}$ and $P$ are constant, we have $|S| = O(\log m)|S^*|$.


\end{document}